\RequirePackage{fix-cm}
\documentclass[smallextended]{svjour3}       
\smartqed  

\usepackage[english]{babel}
\usepackage{enumerate}
\usepackage{color}
\usepackage{subfigure}
\usepackage{dsfont}
\usepackage{graphicx}
\usepackage{epstopdf}
\usepackage{amsmath}
\usepackage{mathtools}
\allowdisplaybreaks
\usepackage{color}

\usepackage{amstext}
\usepackage{amssymb}
\usepackage{mathrsfs}

\usepackage[english]{babel}
\usepackage{amsmath}
\usepackage{amstext}
\usepackage{amssymb}
\usepackage{mathtools}
\usepackage{lipsum}
\usepackage{amsfonts}
\usepackage{graphicx}
\usepackage{epstopdf}
\usepackage{algorithmic}
\ifpdf
  \DeclareGraphicsExtensions{.eps,.pdf,.png,.jpg}
\else
  \DeclareGraphicsExtensions{.eps}
\fi
\usepackage{xcolor}

\def\R{\mathbb{R}}

\def\Gr{\mathop{\rm Gr}}

\def\tJ{\tilde{J}}

\def\B{{\mathcal B}}

\def\C{{\mathcal C}}

\def\P{{\mathcal P}}

\def\M{{\mathcal M}}

\def\bmu{{\boldsymbol \mu}}
\def\bnu{{\boldsymbol \nu}}
\def\bpi{{\boldsymbol \pi}}
\def\bxi{{\boldsymbol \xi}}

\def\tpi{{\tilde \pi}}

\def\tb{{\tilde b}}
\def\hb{{\hat b}}

\def\sPr{{\mathsf{Pr}}}

\def\sX{{\mathsf X}}
\def\sY{{\mathsf Y}}
\def\sA{{\mathsf A}}

\def\sZ{{\mathsf Z}}
\def\sM{{\mathsf M}}
\def\sE{{\mathsf E}}

\def\sG{{\mathsf G}}
\def\sW{{\mathsf W}}
\def\sS{{\mathsf S}}
\def\sK{{\mathsf K}}
\def\sB{{\mathsf B}}
\def\sV{{\mathsf V}}

\newtheorem{definition}{Definition}
\newtheorem{theorem}{Theorem}
\newtheorem{corollary}{Corollary}
\newtheorem{proposition}{Proposition}
\newtheorem{lemma}{Lemma}
\newtheorem{remark}{Remark}
\newtheorem{assumption}{Assumption}
\allowdisplaybreaks

\allowdisplaybreaks

\begin{document}
\title{Partially Observed Discrete-Time Risk-Sensitive Mean Field Games}

\author{Naci Saldi \and Tamer Ba\c{s}ar  \and  Maxim Raginsky}

\institute{Naci Saldi, Corresponding author \at
             Department of Mathematics,    
             Bilkent University \\
              Ankara, Turkey\\
              naci.saldi@bilkent.edu.tr  
           \and
           Tamer Ba\c{s}ar  and  Maxim Raginsky  \at
              Coordinated Science Laboratory, University of  Illinois\\
              Urbana, IL\\
              basar1,maxim@illinois.edu
}

\date{Received: date / Accepted: date}
  
\maketitle

\begin{abstract}
In this paper, we consider discrete-time partially observed mean-field games with the risk-sensitive optimality criterion. We introduce risk-sensitivity behaviour for each agent via an exponential utility function. In the game model, each agent is weakly coupled with the rest of the population through its individual cost and state dynamics via the empirical distribution of states. We establish the mean-field equilibrium in the infinite-population limit using the technique of converting the underlying original partially observed stochastic control problem to a fully observed one on the belief space and the dynamic programming principle. Then, we show that the mean-field equilibrium policy, when adopted by each agent, forms an approximate Nash equilibrium for games with sufficiently many agents. We first consider finite-horizon cost function, and then, discuss extension of the result to infinite-horizon cost in the next-to-last section of the paper.
\end{abstract}

\keywords{Mean field games \and partial observation \and risk sensitive cost.}

\section{Introduction}\label{sec1}

Mean-field games have been introduced in \cite{HuMaCa06} and \cite{LaLi07} to show the existence of approximate Nash equilibria for fully observed non-cooperative continuous time games, when the number of agents is large but finite. The underlying idea of the mean-field method is to transform the decentralized game problem to a centralized stochastic control problem using the so-called `Nash certainty equivalence (NCE) principle' \cite{HuMaCa06}. The optimal solution of this control problem, calibrated appropriately using the empirical distribution of the term that (weakly) couples the players,  provides an approximate Nash equilibrium for games with a sufficiently large number of agents. To obtain the optimal solution to the associated stochastic control problem, one should simultaneously solve a Fokker-Planck equation evolving forward in time and a Hamilton-Jacobi-Bellman equation evolving backward in time. We refer the reader to \cite{HuCaMa07,TeZhBa14,Hua10,BeFrPh13,Ca11,CaDe13,GoSa14,MoBa16} for studies of fully-observed continuous-time mean-field games with different models and cost functions, such as games with major-minor players, risk-sensitive games, games with Markov jump parameters, and LQG games.

In this paper, we study discrete-time partially-observed mean-field games with risk-sensitive optimality criteria. Risk-sensitivity brings in an element of robustness to decision making, and has been widely used in many fields, such as control, economics, financial engineering, and operations research, among others. As opposed to risk-neutral optimization where only the mean value of the cost is considered, risk-sensitive one places positive weights on also the higher moments, thus capturing the risk element (see \cite{TeZhBa14,Whi90,Bas00}). In the model we study in this paper, we have a large but finite number of agents interacting with each other through their individual dynamics and cost functions via the mean-field term (i.e., the empirical distribution of their states). It is known that establishing the existence of Nash equilibria for these types of games is quite difficult due to the (almost) decentralized and noisy nature of the information structure of the problem \cite{Bar78-a,Bar78-b}. Therefore, it is of interest to find an approximate equilibrium with reduced complexity. To that end, upon letting the number of agents go to infinity, the mean-field term converges to the distribution of the state of a single generic agent. This decouples the dynamics and cost functions of the agents from each other, and because of that, in the limiting case, a generic agent is faced with a stochastic control problem with a constraint on the distribution of the state at each time (i.e., a mean-field game problem). The main goal in these problems is to show the existence of a policy and a state distribution flow such that this policy is an optimal solution of the stochastic control problem when the total population behavior is modeled by the state distribution flow and the resulting distribution of each agent's state is same as the state distribution flow when the generic agent applies this policy. This equilibrium condition is called the Nash certainty equivalence (NCE) principle in the literature. In this paper, we first consider the existence of such an equilibrium for the limiting case, and then establish that the policy in this equilibrium constitutes an approximate Nash equilibrium for finite-agent games with sufficiently many agents.

In the literature, \emph{partially-observed} mean-field games have not been studied much, especially in the discrete-time setup. Indeed, this work seems to be the first one that studies discrete-time risk-sensitive mean-field games under partial observations. Prior works have mostly considered the risk-neutral continuous-time setup. It is obvious that analyses of continuous-time and discrete-time setups are quite different, requiring different sets of tools. In \cite{HuCaMa06}, the authors study a partially-observed continuous-time mean-field game with linear individual dynamics. In \cite{SeCa14,SeCa15,SeCa16-3}, the authors consider a continuous-time mean-field game with major-minor agents and nonlinear dynamics where the minor agents can partially observe the state of the major agent. In \cite{SeCa16,SeCa16-2}, the same authors also develop a nonlinear filtering theory for McKean-Vlasov type stochastic differential equations that arise as the infinite population limit of the partially-observed differential game of the mean-field type. In \cite{CaKi17}, the authors study the linear quadratic mean-field game with major-minor agents where the minor agents can partially observe the state of the major agent. In \cite{FiCa15,FiCa19}, the authors consider the linear quadratic mean-field game, again with major-minor agents where, in this case, both the minor agents and the major agent can partially observe the state of the major agent.
In \cite{TaMe16}, the authors study a continuous-time partially observed stochastic control problem of the mean-field type and establish a maximum principle to characterize the optimal control. In \cite{HuWa14}, the authors consider a continuous-time mean-field game with linear individual dynamics where two types of partial information structure are considered: (i) agents cannot observe the white noise which is common to all agents, (ii) agents can access the additive white-noise version of their own states.

For risk-sensitive cost criteria, existing works are mostly on the continuous-time set-up, with \cite{MoBa15}, discussed further below, being one exception. Now, in continuous-time set-up, reference \cite{TeZhBa14} studies a class of mean-field games with non-linear individual dynamics and a risk-sensitive cost function. They characterize the mean-field equilibrium via coupled HJB and FP equations and explicit solutions to these equations are given when the individual state dynamics are linear. In \cite{Tem15}, the author considers a continuous-time mean-field game with nonlinear individual dynamics, where state dynamics have $L^p$-norm structure. Stochastic maximum principle is used to characterize the optimal solution of the problem. In \cite{DjTe16}, the authors study a partially-observed version of the continuous-time risk-sensitive mean-field game. They establish a stochastic maximum principle for the characterization of the mean-field equilibrium. Reference \cite{MoBa17} considers continuous-time risk-sensitive mean-field games with linear individual dynamics and local state information for the players. First a generic risk-sensitive optimal control problem is solved which yields mean-field equilibrium, and then it is shown that the policies in mean-field equilibrium lead to an approximate Nash equilibrium for games with a sufficiently large number of agents. It is also shown that this approximate Nash equilibrium is partially equivalent to the approximate Nash equilibrium of a certain robust mean-field game problem.  Finally, \cite{MoBa15} presents the counterparts of these results for the discrete-time linear-quadratic risk-sensitive mean-field game.

Here, we consider discrete-time mean-field games with Polish state, action, and observation spaces (i.e., complete and separable metric spaces)  under risk-sensitive optimality criteria for the players. In the infinite population limit of such games, a generic agent should solve a partially observed stochastic control problem under the NCE principle. Due to the constraints induced by NCE principle, common techniques used to analyze partially observed stochastic control problems are not sufficient. To establish the existence of an equilibrium solution in the infinite population limit, we have to bring in the fixed-point approach that is used to obtain equilibria in classical game problems, along with the technique of converting partially observed optimal control problems to fully observed ones on the belief space. The definitions of the finite-agent game and the mean-field game problems are given in Section~\ref{sec2} and Section~\ref{sec3}, respectively. In Section~\ref{main-proof}, we prove the existence of a mean-field equilibrium. In Section~\ref{sec4} and Section~\ref{sec4-1}, we establish that the mean-field equilibrium policy is approximately Nash for finite-agent games with sufficiently many agents. In Section~\ref{infinite}, we extend previous results to games with infinite-horizon risk-sensitive cost functions. Section~\ref{conc} concludes the paper.

In an earlier paper \cite{SaBaRa18}, we studied the risk-neutral version of this problem under a similar set of assumptions on the system components. There are some parallels between the techniques used in this paper and those in \cite{SaBaRa18} to show the existence of a mean-field equilibrium and to prove that the policies in mean-field equilibrium provide an approximate Nash equilibrium for games with large but finitely many agents. In this paper, we exploit this connection, and refer the reader to \cite{SaBaRa18} for proofs of certain results. We note, however, that as far as their analyses go, there are considerable technical differences between risk-sensitive and risk-neutral cost functions. The fact that, in the risk-sensitive case, the cost function is in a multiplicative form leads to complication in the analysis of the optimality condition. Therefore, to establish the existence of a mean-field equilibrium in the infinite-population limit and an approximate Nash equilibrium in the finite-agent case, we need to first transform the risk-sensitive problem to one where the cost function is risk-neutral and in an additive form. However, in this risk-neutral form, the one-stage cost function and the transition probability become non-homogeneous (i.e., time-dependent) as opposed to the risk-neutral problem in \cite{SaBaRa18}. Hence, after a careful execution of this step, we can prove the existence of a mean-field equilibrium by adapting the technique developed in \cite{SaBaRa18} to the non-homogeneous and finite-horizon case. We also note that in \cite{SaBaRa18-r} we have studied the fully-observed version of the same problem under a slightly different set of assumptions on the system components. Indeed, to prove the existence of an approximate Nash equilibrium, here we generalize the results established in \cite{SaBaRa18-r} to the game models with expanding state spaces and non-homogeneous system components.


\noindent\textbf{Notation.} For a metric space $\sE$, we let $C_b(\sE)$ denote the set of all bounded continuous real functions on $\sE$, $\P(\sE)$ denote the set of all Borel probability measures on $\sE$, and $\B(\sE)$ denote the collection of Borel sets. For any $\sE$-valued random element $x$, ${\cal L}(x)(\,\cdot\,) \in \P(\sE)$ denotes the distribution of $x$. A sequence $\{\mu_n\}$ of measures on $\sE$ is said to converge weakly to a measure $\mu$ if $\int_{\sE} g(e) \mu_n(de)\rightarrow\int_{\sE} g(e) \mu(de)$ for all $g \in C_b(\sE)$. For any $\nu \in \P(\sE)$ and measurable real function $g$ on $\sE$, we define $\nu(g) = \int g d\nu$. For any subset $B$ of $\sE$, we let $\partial B$ and $B^c$ denote the boundary and complement of $B$, respectively. The notation $v\sim \nu$ means that the random element $v$ has distribution $\nu$. Unless otherwise specified, the term ``measurable" will refer to Borel measurability.

\section{Finite Player Game Model}\label{sec2}

\subsection{Original Game Model}
Let $\sS$, $\sA$, and $\sY$ be Polish spaces. We consider a discrete-time partially-observed $N$-agent mean-field game with a state space $\sS$, an action space $\sA$, and an observation space $\sY$. For every $i \in \{1,2,\ldots,N\}$, the state, the action, and the observation of Agent~$i$ at time $t$ ($t=0,1,2,\ldots$) are respectively denoted by
$s^N_i(t) \in \sS, \text{ } u^N_i(t) \in \sA, \text{ } \text{and} \text{ } g^N_i(t) \in \sY.$
We let
$
d_t^{(N)}(\,\cdot\,) = \frac{1}{N} \sum_{i=1}^N \delta_{s_i^N(t)}(\,\cdot\,) \in \P(\sS) \nonumber
$
denote the empirical distribution of the states (i.e., mean-field term) at time $t$, where $\delta_s\in\P(\sS)$ is the Dirac measure at $s$; that is, $\delta_s(A) = 1$ if $s \in A$ and otherwise $0$.

At the initial time step $t=0$, the states
$(s^N_1(0),\ldots,s^N_N(0)) \sim \kappa_0 \otimes \ldots \otimes \kappa_0$
are independent and identically distributed according to $\kappa_0$. For each $t \ge 0$, the current-observations $(g^N_1(t),\ldots,g^N_N(t))$ and the next-states $(s^N_1(t+1),\ldots,s^N_N(t+1))$ are distributed according to the probability laws
\begin{align}\label{eq:state_spec}
\prod^N_{i=1} l\big(dg^N_i(t)\big|s^N_i(t)\big) \text{  }\text{ and } \text{  }
\prod^N_{i=1} q\big(ds^N_i(t+1)\big|s^N_i(t),u^N_i(t),d^{(N)}_t\big),
\end{align}
where $q : \sS \times \sA \times \P(\sS) \to \P(\sS)$ is the state transition kernel and $l: \sS \to \P(\sY)$ is the observation kernel. Note that the state dynamics of each agent are weakly coupled through the mean-field term $d^{(N)}_t$.

For any Agent~$i$, define the history spaces $\sG_0 = \sY$ and $\sG_{t}= (\sY\times\sA)^t\times\sY$ for $t=1,2,\ldots$, all endowed with product Borel $\sigma$-algebras. A \emph{policy} for Agent~$i$ is a sequence $\pi^i=\{\pi_{t}^i\}$ of stochastic kernels on $\sA$ given $\sG_{t}$; that is, for any $t\geq0$, 
$
u_i^N(t) \sim \pi_t^i(\cdot|\gamma_i^N(t)), 
$
where $\gamma^N_i(t) = \big(g^N_i(t),u^N_i(t-1),g^N_i(t-1)\ldots,u^N_i(0),g^N_i(0)\big)$
is the observation-action history observed by Agent~$i$ up to time $t$. The set of all policies for Agent~$i$ is denoted by $\Pi_i$. Let $\tilde{\Pi}_i$ be the set of policies in $\Pi_i$ which only use the observations; that is, $\pi \in \tilde{\Pi}_i$ if $\pi_t: \prod_{k=0}^t \sY \rightarrow \P(\sA)$ for each $t\geq0$. Let
${\bf \Pi}^{(N)} = \prod_{i=1}^N \Pi_i \text{ } \text{and} \text{ } \tilde{{\bf \Pi}}^{(N)} = \prod_{i=1}^N \tilde{\Pi}_i.$
We let ${\boldsymbol \pi}^{(N)} = (\pi^1,\ldots,\pi^N)$ ($\pi^i \in \Pi_i$) denote the $N$-tuple of joint policies of all the agents in the game. Under such an $N$-tuple of policies, the actions of agents at each time $t \ge 0$ are obtained with respect to the conditional probability distribution
\begin{align}\label{eq:policy_spec}
\prod^N_{i=1} \pi^i_t\big(du^N_i(t)\big|\gamma^N_i(t)\big).
\end{align}

The \emph{one-stage cost} function for a generic agent is a measurable function $m : \sS \times \sA \times \P(\sS) \to [0,\infty)$. Then, the agent's finite-horizon \emph{risk-sensitive} cost under a policy ${\boldsymbol \pi}^{(N)} \in {\bf \Pi}^{(N)}$ is given by
\begin{align}
V_i^{(N)}({\boldsymbol \pi}^{(N)}) &= \frac{1}{\lambda} \log\biggl( E^{{\boldsymbol \pi}^{(N)}}\biggl[ e^{\lambda\sum_{t=0}^{T}\beta^{t}m(s_{i}^N(t),u_{i}^N(t),d^{(N)}_t)}\biggr]\biggr), \nonumber
\end{align}
where $\beta \in (0,1]$ is the discount factor, $\lambda > 0$ is the risk factor, and $T$ is the finite horizon of the problem. Here, $E^{{\boldsymbol \pi}^{(N)}}\big[\cdot\big]$ denotes the expectation with respect to the probability law, which is uniquely specified by the kernels in \eqref{eq:state_spec} and \eqref{eq:policy_spec} and the initial state distribution $\kappa_0$.

Since $\frac{1}{\lambda}\log(\cdot)$ is a strictly increasing function, without loss of generality, it suffices to consider only the part with expectation:
\begin{align}
W_i^{(N)}({\boldsymbol \pi}^{(N)}) &= E^{{\boldsymbol \pi}^{(N)}}\biggl[ e^{\lambda\sum_{t=0}^{T}\beta^{t}m(s_{i}^N(t),u_{i}^N(t),d^{(N)}_t)}\biggr]. \nonumber
\end{align}
With this cost function, the equilibrium solution for the game is defined as follows:
\begin{definition}
A policy ${\boldsymbol \pi}^{(N*)}= (\pi^{1*},\ldots,\pi^{N*})$ constitutes a \emph{Nash equilibrium} for the $N$-player game, if
\begin{align}
W_i^{(N)}({\boldsymbol \pi}^{(N*)}) = \inf_{\pi^i \in \Pi_i} W_i^{(N)}({\boldsymbol \pi}^{(N*)}_{-i},\pi^i) \nonumber
\end{align}
for each $i=1,\ldots,N$, where ${\boldsymbol \pi}^{(N*)}_{-i} = (\pi^{j*})_{j\neq i}$.
\end{definition}

As we have explained in detail in \cite{SaBaRa18}, establishing the existence of Nash equilibria for \emph{partially-observed} mean-field games is challenging due to the (almost) decentralized and noisy nature of the information structure of the problem. To that end, we slightly change the definition of Nash equilibrium in this model and adopt the approximate Nash equilibrium concept instead of exact Nash equilibrium.
\begin{definition}\label{def1}
A policy ${\boldsymbol \pi}^{(N*)} \in \tilde{{\bf \Pi}}^{(N)}$ is a \emph{Nash equilibrium} if
\begin{align*}
W_i^{(N)}({\boldsymbol \pi}^{(N*)}) &= \inf_{\pi^i \in \tilde{\Pi}_i} W_i^{(N)}({\boldsymbol \pi}^{(N*)}_{-i},\pi^i)
\end{align*}
for each $i=1,\ldots,N$, and an \emph{$\varepsilon$-Nash equilibrium} (for a given $\varepsilon > 0$) if
\begin{align*}
W_i^{(N)}({\boldsymbol \pi}^{(N*)}) &\leq \inf_{\pi^i \in \tilde{\Pi}_i} W_i^{(N)}({\boldsymbol \pi}^{(N*)}_{-i},\pi^i) + \varepsilon
\end{align*}
for each $i=1,\ldots,N$.
\end{definition}

According to this definition, the agents can only use their local observations $(g^N_i(t),\ldots,g^N_i(0))$ to construct their policies. In real life applications, agents typically have access only to their local observations. Hence, it suffices to establish the existence of an approximate Nash equilibrium for the game with a local information structure. In addition, in the discrete-time mean field literature, it is common to establish the existence of approximate Nash equilibria with local (decentralized) information structures (see \cite{AdJoWe15} \cite{Bis15}). This is true for partially observed case as well (see \cite{SaBaRa18}).

Here, our goal is to establish the existence of approximate Nash equilibria for games with sufficiently many agents. Indeed, if the number of agents is small, it is all but impossible to show even the existence of approximate Nash equilibria for these types of games. Therefore, it is key to assume that the number of agents is large (but finite). With this assumption, we can go to the infinite population limit, for which we can model the mean-field term as an exogenous state-measure flow, which should be consistent with the distribution of a generic agent (i.e., the NCE principle) by the law of large numbers. In this case, to establish the existence of an equilibrium, a generic agent should solve a classical partially observed stochastic control problem with a constraint on the distributions on the states (i.e., mean-field game). Then, we expect that if each agent in the finite-agent $N$ game adopts the equilibrium policy in the infinite-population limit, the resulting policy will be an approximate Nash equilibrium for all sufficiently large $N$.

Our approach to prove the existence of approximate Nash equilibria can be summarized as follows: (i) Note that, in the risk-sensitive criteria, the one-stage cost functions are in a multiplicative form as opposed to the risk-neutral setting. As stated earlier, this makes the analysis of the problem quite complicated.  Therefore, we first construct an equivalent non-homogeneous game model, where the cost can be written in an additive form as in the risk-neutral case (see Section~\ref{sub1sec2}). (ii) Then, we introduce the infinite-population limit ($N \to \infty$) of the equivalent game model to approximate the finite-agent setting (see Section~\ref{sec3}). (iii) By adapting the proof technique in \cite{SaBaRa18} to the non-homogeneous and finite-horizon set-up, we prove the existence of an appropriately defined mean-field equilibrium for this limiting infinite-population game (see Section~\ref{main-proof}). (iv) Then, we return to the finite-$N$ case for the equivalent game model and show that, if each agent in the game problem adopts the mean-field equilibrium policy, then the resulting policy will be an approximate Nash equilibrium for all sufficiently large $N$. Since the equivalent game model is identical to the original game model in terms of cost functions, this establishes the existence of approximate Nash equilibria for the original game model (see Sections~\ref{sec4} and \ref{sec4-1}).

Now, proceeding along the lines above, we first introduce the following assumptions, imposed throughout the paper.

\begin{assumption}\label{as1}
\begin{itemize}
\item [ ]
\item [(a)] The cost function $m$ is bounded and continuous with $\| m \| = \sup_{s \in \sS} |m(s)| \leq K$.
\item [(b)] The stochastic kernel $q$ is weakly continuous in $(s,u,\kappa)$; i.e., \newline $q(\,\cdot\,|s(k),u(k),\kappa_k) \rightarrow q(\,\cdot\,|s,u,\kappa)$ weakly when $(s(k),u(k),\kappa_k) \rightarrow (s,u,\kappa)$.
\item [(c)] The observation kernel $l$ is continuous in $s$ with respect to total variation norm; i.e., for all $s$, $l(\,\cdot\,|s_k) \rightarrow l(\,\cdot\,|s)$ in total variation norm when $s_k \rightarrow s$.
\item [(d)] $\sA$ is compact.
\item [(e)] There exist a constant $\alpha \ge 0$  and a continuous moment function $v: \sS \rightarrow [1,\infty)$  (see \cite[Definition E.7]{HeLa96}) such that
\begin{align}
\sup_{(u,\kappa) \in \sA \times \P(\sS)} \int_{\sS} v(y) q(dy|s,u,\kappa) \leq \alpha v(s).
\end{align}
\item [(f)] The initial probability measure $\kappa_0$ satisfies
$
\int_{\sS} v(s) \kappa_0(ds) = M < \infty. \nonumber
$
\end{itemize}
\end{assumption}

\subsection{Equivalent Game Model}\label{sub1sec2}

In this section, we construct an equivalent game model whose states are the states of the original model plus the one-stage costs incurred up to that time. Namely, the state at time $t$ for Agent~$i$ is
\begin{align}
x_i^N(t) =  \biggl(s_i^N(t),\sum_{k=0}^{t-1} \beta^k m(s_i^N(k),u_i^N(k),d_k^{(N)})\biggr). \nonumber
\end{align}
In this new model, finite-horizon risk-sensitive cost function can be written in an additive-form like in risk-neutral case. For this new game model, we have been inspired by \cite{BaRi14}, in which the authors study the classical fully-observed risk-sensitive control problem. For a generic agent, this new game model is specified by
\begin{align}
\biggl( \sX, \sA, \sY, \{p_t\}_{t = 0}^{T+1}, r, \{c_t\}_{t=0}^{T+1}, \mu_0 \biggr), \nonumber
\end{align}
where
$
\sX = \sS \times [0,L] \nonumber
$
is the new state space with $L = \frac{K}{1-\beta}$, where $L$ is the maximum risk-neutral discounted-cost that can be incurred. For every $t$, the state transition kernel $p_t : \sX \times \sA \times \P(\sX) \to \P(\sX)$ is defined as\footnote{In the remainder of this paper, we use letter `$a$' instead of `$u$', to denote actions, to emphasize that they are generated using the new game model.}:
\begin{align}
p_t\bigl(B \times D \big| x(t),a(t),\mu_t\bigr) = q(B|s(t),a(t),\mu_{t,1}) \otimes \delta_{m(t) + \beta^t m(s(t),a(t),\mu_{t,1})}(D), \nonumber
\end{align}
where $B \in \B(\sS)$, $D \in \B([0,L])$,
$
x(t) = (s(t),m(t)), \nonumber
$
and $\mu_{t,1}$ is the marginal of $\mu_t$ on $\sS$. Here, $p_t$ is indeed the controlled transition probability of the next state $s_i^N(t+1)$ and current risk-neutral total discounted cost $$\sum_{k=0}^{t} \beta^k m(s_i^N(k),a_i^N(k),d_k^{(N)})$$ given the current state-action pair $(s_i^N(t),a_i^N(t))$ and past risk-neutral total discounted cost $\sum_{k=0}^{t-1} \beta^k m(s_i^N(k),a_i^N(k),d_k^{(N)})$ in the original game. The observation kernel $r: \sX \rightarrow \P(\sY)$ is equivalent to the observation kernel $l$ in the original problem; that is, $r(dy|x) = l(dy|s)$ where $x = (s,m)$. For each $t$, the one-stage cost function $c_t: \sX \times \sA \times \P(\sX) \rightarrow [0,\infty)$ is defined as:
\begin{align}
c_t(x(t),a(t),\mu_t) =
\begin{cases}
0, & \text{ } \text{ if $t \leq T$} \\
e^{\lambda m(t)}, & \text{ } \text{ if $t = T + 1$}.
\end{cases}\nonumber
\end{align}
Finally, the initial measure $\mu_0$ is given by $\mu_0(dx(0)) = \kappa_0(ds(0)) \otimes \delta_0(dm(0))$, where the initial states $\{x_i^N(0)\}$ are independent and identically distributed according to $\mu_0$. Note that, in this equivalent game model, the finite-horizon is $T+1$ instead of $T$ and system components depend on time $t$. We also define the empirical distribution of the states at time $t$ as follows:
\begin{align}
e_t^{(N)}(\,\cdot\,) = \frac{1}{N} \sum_{i=1}^N \delta_{x_i^N(t)}(\,\cdot\,) \in \P(\sX). \nonumber
\end{align}

Suppose that Assumption~\ref{as1} holds. Then, for each $t$, the following are true for the new game model:
\begin{itemize}
\item [(I)] The one-stage cost function $c_t$ is bounded and continuous.
\item [(II)] The stochastic kernel $p_t$ is weakly continuous.
\item [(III)] The observation kernel $r$ is continuous with respect to the total variation distance.
\item [(IV)] Let $w: \sX \rightarrow [1,\infty)$ be defined as $w(x) = w((s,m)) = v(s)$, which is a moment function. Then, we have
\begin{align}
\sup_{(a,\mu) \in \sA \times \P(\sX)} \int_{\sX} w(y) p_t(dy|x,a,\mu) \leq \alpha w(x).
\end{align}
\item [(V)] The initial probability measure $\mu_0$ satisfies
$
\int_{\sX} w(x) \mu_0(dx) = M < \infty. \nonumber
$
\end{itemize}

Recall that $\tilde{\Pi}_i$ denotes the set of policies for Agent $i$ that only use observations in the original game. Note that $\tilde{\Pi}_i$ is also the set of policies for Agent $i$ that only use observations in the new game model. For Agent~$i$, the finite-horizon risk-neutral total cost under the $N$-tuple of policies ${\boldsymbol \pi}^{(N)} \in {\bf \tilde{\Pi}}^{(N)}$ is denoted as $J_i^{(N)}({\boldsymbol \pi}^{(N)})$; that is
\begin{align}
J_i^{(N)}({\boldsymbol \pi}^{(N)}) &= E^{{\boldsymbol \pi}^{(N)}}\biggl[ \sum_{t=0}^{T+1} c_t(x_i^N(t),a_i^N(t),e_t^{(N)})\biggr]. \nonumber
\end{align}
The following proposition makes the connection between this new model and the original model. The proof is straightforward, and so, we omit the details (see the proof of \cite[Proposition 5.1]{SaBaRa18-r}). 

\begin{proposition}\label{app-prop1}
For any ${\boldsymbol \pi}^{(N)} \in {\bf \tilde{\Pi}}^{(N)}$ and $i=1,\ldots,N$, we have $J_i^{(N)}({\boldsymbol \pi}^{(N)}) = W_i^{(N)}({\boldsymbol \pi}^{(N)})$.
\end{proposition}


Proposition~\ref{app-prop1} states that the new game model is equivalent to the original game model in terms of cost functions. This is true because the new game model consists of the one-stage costs incurred up to the current time as an additional state variable. Therefore, if we take the exponent of this additional state at time $T+1$ as in the definition of $c_{T+1}$, we obtain the risk-sensitive cost of the original game model. Hence, in the remainder of this paper, we replace the original game model with the new one; that is, from this point on, we have the following system components satisfying (I)-(V):
\begin{align}
\biggl( \sX, \sA, \sY, \{p_t\}_{t = 0}^{T+1}, r, \{c_t\}_{t=0}^{T+1}, \mu_0 \biggr). \nonumber
\end{align}

\begin{remark}
Note that in the new game model, the time horizon is $T+1$, which means that agents should also design control policies for the time step $T+1$. However, note that control policies at time step $T+1$ do not affect the cost function (i.e., one-stage cost at time $T+1$ is only a function of the state), and thus agents indeed do not need to select these policies in the new game model. Hence, we can in a sense view the time horizons of the two problems as $T$. 
\end{remark}

Note that the cost functions $J_i^{(N)}({\boldsymbol \pi}^{(N)})$ of this new game model are in additive form (i.e., risk-neutral). Therefore, we can use a technique similar to the one in \cite{SaBaRa18} to prove the existence of an approximate Nash equilibrium. To this end, we will first consider the infinite-population limit of the new game model and prove the existence of an equilibrium. Then, we will go back to the finite agent case and establish the existence of approximate Nash equilibrium for the new game model using the infinite population equilibrium solution. Since, by Proposition~\ref{app-prop1}, the new game model has the same cost function as the original game model, the last result also implies the existence of an approximate Nash equilibrium for the original game, which was the main goal of this paper.

\section{Partially observed mean-field games and mean-field equilibria}\label{sec3}

In this section, we introduce the infinite population limit of the new game introduced in the preceding section. Although it is called mean-field game, it is not game in the classical sense: it is a stochastic control problem whose state distribution at each time step should satisfy a certain consistency condition. The optimal solution of this problem is referred to as mean-field equilibrium. In other words, we have a single agent and model the mean-field term by an exogenous \textit{state-measure flow} $\bmu := (\mu_t)_{t = 0}^{T+1} \subset \P(\sX)$  with a given initial condition $\mu_0$, by the law of large numbers. This measure flow $\bmu$ should also be consistent with the state distributions of this single agent when the agent acts optimally. The precise mathematical description of the problem is given as follows.

The mean-field game model for a generic agent is specified by
\begin{align}
\biggl( \sX, \sA, \sY, \{p_t\}_{t = 0}^{T+1}, r, \{c_t\}_{t=0}^{T+1}, \mu_0 \biggr), \nonumber
\end{align}
where, as before, $\sX$, $\sA$, and $\sY$ are the state, action, and observation spaces, respectively. The stochastic kernel $p_t : \sX \times \sA \times \P(\sX) \to \P(\sX)$ denotes the transition probability, and $r: \sX \times \P(\sX) \to \P(\sY)$ denotes the observation kernel. The measurable function $c_t: \sX \times \sA \times \P(\sX) \rightarrow [0,\infty)$ is the one-stage cost function and $\mu_0$ is the distribution of the initial state.

Recall the history spaces $\sG_0 = \sY$ and $\sG_{t}=(\sY\times\sA)^{t}\times \sY$ for $t=1,2,\ldots$, all endowed with product Borel $\sigma$-algebras. A \emph{policy} is a sequence $\pi=\{\pi_{t}\}$ of stochastic kernels on $\sA$ given $\sG_{t}$. The set of all policies is denoted by $\Pi$.

We let $\M = \bigl\{\bmu \in \P(\sX)^{T+2}: \mu_0 \text{ is fixed}\bigr\}$ be the set of all state-measure flows with a given initial condition $\mu_0$. Given any measure flow $\bmu \in \M$, the evolution of the states, observations, and actions is as follows
\begin{align}
x(0) &\sim \mu_0, \nonumber \\
y(t) &\sim r(\,\cdot\,|x(t)), \text{ } t=0,1,\ldots \nonumber \\
x(t) &\sim p_{t-1}(\,\cdot\,|x(t-1),a(t-1),\mu_{t-1}), \text{ } t=1,2,\ldots \nonumber \\
a(t) &\sim \pi_t(\,\cdot\,|\gamma(t)), \text{ } t=0,1,\ldots, \nonumber
\end{align}
where $\gamma(t) \in \sG_t$ is the observation-action history  up to time $t$. An initial distribution $\mu_0$ on $\sX$, a policy $\pi$, and a state-measure flow $\bmu$ define a unique probability measure $P^{\pi}$ on $(\sX \times \sY \times \sA)^{T+2}$. The expectation with respect to $P^{\pi}$ is denoted by $E^{\pi}[\,\cdot\,]$. A policy $\pi^{*} \in \Pi$ is said to be optimal for $\bmu$ if
$
J_{\bmu}(\pi^{*}) = \inf_{\pi \in \Pi} J_{\bmu}(\pi), \nonumber
$
where the finite-horizon cost of policy $\pi$ with measure flow $\bmu$ is given by
\begin{align}
J_{\bmu}(\pi) &=  E^{\pi}\biggl[ \sum_{t=0}^{T+1} c_t(x(t),a(t),\mu_t) \biggr] \nonumber
\end{align}

Using these definitions, we first define the set-valued mapping
$
\Psi : \M \rightarrow 2^{\Pi} \nonumber
$
as $\Psi({\boldsymbol \mu}) = \{\pi \in \Pi: \pi \text{ is optimal for } {\boldsymbol \mu}\}$.
Conversely, we define a single-valued mapping
$
\Lambda : \Pi \to \M \nonumber
$
as follows: given $\pi \in \Pi$, the state-measure flow $\bmu := \Lambda(\pi)$ is constructed recursively as
\begin{align}
\mu_{t+1}(\,\cdot\,) = \int_{\sX \times \sA} p_t(\,\cdot\,|x(t),a(t),\mu_t) P^{\pi}(da(t)|x(t)) \mu_t(dx(t)), \nonumber
\end{align}
where $P^{\pi}(da(t)|x(t))$ denotes the conditional distribution of $a(t)$ given $x(t)$ under $\pi$ and $(\mu_{\tau})_{0\leq\tau\leq t}$. Using $\Psi$ and $\Lambda$, we now introduce the mean-field equilibrium.

\begin{definition}
A pair $(\pi^*,{\boldsymbol \mu}^*) \in \Pi \times \M$ is a \emph{mean-field equilibrium} if $\pi^* \in \Psi({\boldsymbol \mu}^*)$ and $\bmu^* = \Lambda(\pi^*)$. 
\end{definition}

\noindent The main result of this section is the existence of a mean-field equilibrium. Later we will show that this mean-field equilibrium constitutes an approximate Nash equilibrium for games with sufficiently many agents.

\begin{theorem}\label{thm:MFE}
The mean-field game
$\bigl(\sX, \sA, \sY, \{p_t\}_{t = 0}^{T+1}, r, \{c_t\}_{t=0}^{T+1}, \mu_0 \bigr)$
admits a mean-field equilibrium $(\pi^*,\bmu^*)$.
\end{theorem}

The proof of Theorem~\ref{thm:MFE} is given in Section~\ref{main-proof}.
Our approach to prove Theorem~\ref{thm:MFE} can be summarized as follows: (i)  first, we lift the partially
observed stochastic control problem a generic agent is faced with for a given measure flow to a fully observed stochastic control problem; (ii) we then transform the fixed
point equation $\pi \in \Psi(\Lambda(\pi))$ characterizing the mean-field equilibrium into a fixed point equation
of a set-valued mapping from the set of state-action measure flows into itself using the Bellman optimality operator; (iii) then, we prove that this set-valued mapping has a closed graph; and (iv) finally, we deduce the existence of a mean-field equilibrium using Kakutani's fixed point theorem.

\section{Proof of Theorem~\ref{thm:MFE}}\label{main-proof}

Note that any measure flow $\bmu \in \M$ leads to a non-homogenous partially-observed Markov decision process (POMDP). Hence, before starting the proof of Theorem~\ref{thm:MFE}, we first review a few relevant results on POMDPs. To this end, fix any $\bmu \in \M$ and consider the corresponding optimal control problem.

Let $\P_w(\sX) = \bigl\{\mu \in \P(\sX): \int_{\sX} w(x) \mu(dx) < \infty \bigr\}$. It is known that any POMDP can be reduced to a (completely observable) MDP (see \cite{Yus76}, \cite{Rhe74}), whose states are the posterior state distributions or beliefs of the observer; that is, the state at time $t$ is
\begin{align}
z(t) = \sPr\{x(t) \in \,\cdot\, | y(0),\ldots,y(t), a(0), \ldots, a(t-1)\} \in \P(\sX). \nonumber
\end{align}
We call this equivalent MDP the belief-state MDP. Note that since ${\cal L}(x(t)) \in \P_w(\sX)$ under any policy by (IV)-(V), we have $$\sPr\{x(t) \in \,\cdot\, | y(0),\ldots,y(t), a(0), \ldots, a(t-1)\} \in \P_w(\sX)$$ almost everywhere. Therefore, the belief-state MDP has state space $\sZ = \P_w(\sX)$ and action space $\sA$. Here, $\sZ$ is endowed with the Borel $\sigma$-algebra generated by the topology of weak convergence. Next, we construct the transition probabilities $\{\eta_t\}_{t=0}^{T+1}$ of the belief-state MDP (see also \cite{Her89}). Let $z$ denote the generic state variable for the belief-state MDP. Fix any $t$. First consider the transition probability on $\sX \times \sY$ given $\sZ \times \sA$
\begin{align}
R_t(x \in A, y \in B|z,a) = \int_{\sX} \kappa_t(A,B|x',a) z(dx'), \nonumber
\end{align}
where $\kappa_t(dx,dy|x',a) = r(dy|x) \otimes p_t(dx|x',a,\mu_t)$. Let us disintegrate $R_t$ as follows
$
R_t(dx,dy|z,a) = H_t(dy|z,a) \otimes F_t(dx|z,a,y). \nonumber
$
Then, we define the mapping $F_t: \sZ \times \sA \times \sY \rightarrow \sZ$ as
\begin{align}
F_t(z,a,y)(\,\cdot\,) = F_t(\,\cdot\,|z,a,y) \label{eq:non_filtering}.
\end{align}
Then, $\eta_t: \sZ \times \sA \rightarrow \P(\sZ)$ is defined as
\begin{align}
&\eta_t(\,\cdot\,|z(t),a(t)) = \int_{\sY} \delta_{F_t(z(t),a(t),y(t+1))}(\,\cdot\,) \text{ } H_t(dy(t+1)|z(t),a(t)). \nonumber
\end{align}
The initial point for the belief-state MDP is $\mu_0$; that is, ${\cal L}(z(0)) \sim \delta_{\mu_0}$. Finally, for each $t$, the one-stage cost function $C_t$ of the belief-state MDP is given by
\begin{align}
C_t(z,a) = \int_{\sX} c_t(x,a,\mu_t) z(dx). \label{eq8}
\end{align}
Hence, the belief-state MDP is a Markov decision process with the components
$
\bigl(\sZ,\sA,\{\eta_t\}_{t=0}^{T+1},\{C_t\}_{t=0}^{T+1},\delta_{\mu_0}\bigr). \nonumber
$

For the belief-state MDP define the history spaces $\sK_0 = \sZ$ and $\sK_{t}=(\sZ\times\sA)^{t}\times\sZ$, $t=1,2,\ldots$. A \emph{policy} is a sequence $\varphi=\{\varphi_{t}\}$ of stochastic kernels on $\sA$ given $\sK_{t}$. The set of all policies is denoted by $\Phi$. A \emph{Markov} policy is a sequence $\varphi=\{\varphi_{t}\}$ of stochastic kernels on $\sA$ given $\sZ$. The set of Markov policies is denoted by $\sM$. Let $\tJ(\varphi,\mu_0)$ denote the finite-horizon cost function of policy $\varphi \in \Phi$ for initial point $\mu_0$ of the belief-state MDP. Notice that any history vector $s(t) = (z(0),\ldots,z(t),a(0),\ldots,a(t-1))$ of the belief-state MDP is a function of the history vector $\gamma(t) = (y(0),\ldots,y(t),a(0),\ldots,a(t-1))$ of the POMDP. Let us write this relation as $i(\gamma(t)) = s(t)$. Hence, for a policy $\varphi = \{\varphi_t\} \in \Phi$, we can define a policy $\pi^{\varphi} = \{\pi_t^{\varphi}\} \in \Pi$ as
$
\pi_t^{\varphi}(\,\cdot\,|\gamma(t)) = \varphi_t(\,\cdot\,|i(\gamma(t))). \nonumber
$
Let us write this as a mapping from $\Phi$ to $\Pi$: $\Phi \ni \varphi \mapsto i(\varphi) = \pi^{\varphi} \in \Pi$. It is straightforward to show that the cost functions $\tJ(\varphi,\mu_0)$ and $J_{\bmu}(\pi^{\varphi})$ are the same. One can also prove that (see \cite{Yus76}, \cite{Rhe74})
\begin{align}
\inf_{\varphi \in \Phi} \tJ(\varphi,\mu_0) &= \inf_{\pi \in \Pi} J_{\bmu}(\pi) \label{eq7}
\end{align}
and furthermore, that if $\varphi$ is an optimal policy for belief-state MDP, then $\pi^{\varphi}$ is optimal for the POMDP as well. Therefore, the optimal control problem for the mean-field game is equivalent to the optimal control of belief-state MDP.

We now derive the conditions that are satisfied by belief-state MDP. To that end, define $W:\sZ \rightarrow \R$ as
\begin{align}
W(z) = \int_{\sX} w(x) z(dx). \nonumber
\end{align}
Note that $W$ is a lower semi-continuous moment function on $\sZ$. One can prove that (see \cite[Section 4]{SaBaRa18}) the belief-state MDP satisfies the following conditions under Assumption~1:

\begin{itemize}
\item [(i)] The cost functions $\{C_t\}$ are bounded and continuous.
\item [(ii)] The stochastic kernels $\{\eta_t\}$ are weakly continuous.
\item [(iii)] $\sA$ is compact and $\sZ$ is $\sigma$-compact.
\item [(iv)] There exists a constant $\alpha \ge 0$ such that
\begin{align}
\sup_{a \in \sA} \int_{\sZ} W(y) \eta_t(dy|z,a) \leq \alpha W(z), \text{ } \text{for all $t$.} \nonumber
\end{align}
\item [(v)] The initial probability measure $\delta_{\mu_0}$ satisfies
$
W(\delta_{\mu_0}) = M < \infty. \nonumber
$
\end{itemize}

\smallskip

With these conditions, we are now ready to prove Theorem~\ref{thm:MFE} by adapting techniques in \cite{SaBaRa18} to the non-homogeneous and finite-horizon set-up.

We first define the mapping $\sB: \P(\sZ) \rightarrow \P(\sX)$, which will define the relation between state-measure flows in the mean-field game and state-measure flows in the belief-state MDP, as follows:
\begin{align}
\sB(\nu)(\,\cdot\,) = \int_{\sZ} z(\,\cdot\,) \text{ } \nu(dz). \nonumber
\end{align}
Using this definition, for any $\bnu \in \P(\sZ \times \sA)^{T+2}$, we define the measure flow $\bmu^{\bnu} \in \P(\sX)^{T+2}$ as follows:
\begin{align}
\bmu^{\bnu} = \bigl(\sB(\nu_{t,1})\bigr)_{t=0}^{T+1}, \nonumber
\end{align}
where for any $\nu \in \P(\sZ \times \sA)$, we let $\nu_1$ denote the marginal of $\nu$ on $\sZ$. Let $\{\eta_t^{\bnu}\}_{t=0}^{T+1}$ and $\{C_t^{\bnu}\}_{t=0}^{T+1}$ be, respectively, the transition probabilities and one-stage cost functions of belief-state MDP induced by the measure flow $\bmu^{\bnu}$. We let $J_{*,t}^{\bnu}: \sZ \rightarrow [0,\infty)$ denote the optimal value function at time $t$ of this belief-state MDP; that is,
\begin{align}
J_{*,t}^{\bnu}(z) = \inf_{\varphi \in \Phi} E^{\varphi} \biggl[ \sum_{k=t}^{T+1} C_k^{\bnu}(z(k),a(k)) \bigg| z(t) = z \biggr]. \nonumber
\end{align}
Let $J_{*}^{\bnu} = \bigl( J^{\bnu}_{*,t}\bigr)_{t=0}^{T+1}$.

To prove the existence of a mean-field equilibrium, we use the technique in  \cite{JoRo88}. To that end, we first transform the fixed point equation $\pi \in \Psi(\Lambda(\pi))$ characterizing the mean-field equilibrium into a fixed-point equation of a set-valued mapping from the set of state-action measure flows $\P(\sZ \times \sA)^{T+2}$ into itself. Then, using Kakutani's fixed point theorem (\cite[Corollary 17.55]{AlBo06}), we deduce the existence of a mean-field equilibrium.

For any  $t$, the \emph{Bellman optimality operator } $T_t^{\bnu}: C_{b}(\sZ)\rightarrow C_{b}(\sZ)$ is given by
\begin{align}
T_t^{\bnu} u(z) = \min_{a \in \sA} \biggl[ C^{\bnu}_t(z,a) + \int_{\sZ} u(y) \eta^{\bnu}_t(dy|z,a) \biggr]. \nonumber
\end{align}
Note that $T_t^{\bnu} J^{\bnu}_{*,t+1} = J^{\bnu}_{*,t}$ for every $t$. The following theorem is a known result in the theory of nonhomogeneous Markov decision processes (see \cite[Theorems 14.4 and 17.1]{Hin70}). For any given $\bnu$, it characterizes the optimal policy of the belief-state MDP.

\begin{theorem}\label{theorem1}
For any $\bnu$, a policy $\varphi \in \sM$ is optimal if and only if, for all $t$,
\begin{align}
&\nu_t^{\varphi} \biggl( \biggr\{ (z,a) : C^{\bnu}_t(z,a) + \int_{\sZ} J_{*,t+1}^{\bnu}(y) \eta^{\bnu}_t(dy|z,a) = T_t^{\bnu} J_{*,t+1}^{\bnu}(z) \biggr\} \biggr) = 1, \label{eq5}
\end{align}
where $\nu_t^{\varphi} = {\cal L}\bigl( z(t),a(t) \bigr)$ under $\varphi$ and $\bnu$.
\end{theorem}

Using Theorem~\ref{theorem1}, we now define the set-valued map from $\P(\sZ \times \sA)^{T+2}$ into itself. To that end, for any $\bnu \in \P(\sZ \times \sA)^{T+2}$, let us define the following sets:
\begin{align}
C(\bnu) &= \biggl\{ \bnu' \in \P(\sZ \times \sA)^{T+2}: \nu'_{0,1} = \delta_{\mu_0},  \, \nu'_{t+1,1}(\,\cdot\,) = \int_{\sZ \times \sA} \eta_t^{\bnu}(\,\cdot\,|z,a) \nu_t(dz,da)\biggr\} \nonumber \\
\intertext{and}
B(\bnu) &= \biggl\{ \bnu' \in \P(\sZ \times \sA)^{T+2}: \forall 0\leq t \leq T+1, \text{ } \nonumber \\
&\phantom{xx}\nu_t' \biggl( \biggr\{ (z,a) : C_t^{\bnu}(z,a) + \int_{\sZ} J_{*,t+1}^{\bnu}(y) \eta_t^{\bnu}(dy|z,a) = T_t^{\bnu} J^{\bnu}_{*,t+1}(z) \biggr\} \biggr) = 1 \biggr\}. \nonumber
\end{align}
Here, the set $C(\bnu)$ characterizes the consistency of the mean-field term with the state distribution of a generic agent, and the set $B(\bnu)$ characterizes optimality of the policy for the mean-field term. The set-valued mapping $\Gamma: \P(\sZ \times \sA)^{T+2} \rightarrow 2^{\P(\sZ \times \sA)^{T+2}}$ is given as follows:
\begin{align}
\Gamma(\bnu) = C(\bnu) \cap B(\bnu). \nonumber
\end{align}
Note that the fixed-point equation $\pi \in \Psi(\Lambda(\pi))$ characterizes the behaviour of the state distribution and the control law in mean-field equilibrium separately. 
To establish the existence of mean-field equilibrium via Kakutani's Fixed Point Theorem or Banach Fixed Point Theorem using this equation, one needs to put some topology on the policy space.
However, by combining the state distribution with the control law, which gives the joint distribution of the state and the action, we can characterize via the set-valued mapping $\Gamma$ the behaviour of the state and the control law together in mean-field equilibrium. This will enable us to deduce the existence of a mean-field equilibrium without introducing a topology for the control laws, which is in general the solution technique in continuous time setup (see \cite{HuMaCa06}).

An element $\bnu$ is a fixed point of $\Gamma$ if $\bnu \in \Gamma(\bnu)$. The following proposition makes the connection between mean-field equilibria and fixed points of $\Gamma$.

\begin{proposition}\label{prop1}
Suppose that $\Gamma$ has a fixed point $\bnu = (\nu_t)_{t = 0}^{T+1}$. Construct a Markov policy $\varphi = \{\varphi_t\}$ for belief-state MDP by disintegrating each $\nu_t$ as $\nu_t(dz,da) = \nu_{t,1}(dz) \varphi_t(da|z)$. Let $\pi^*=\pi^{\varphi}$ and $\bmu^* = (\sB(\nu_{t,1}))_{t = 0}^{T+1}$. Then the pair $(\pi^{*},\bmu^*)$ is a mean-field equilibrium.
\end{proposition}

\begin{proof}
Note that, since $\bnu \in C(\bnu)$, we have $\nu_{t} = {\cal L}\bigl( z(t),a(t) \bigr)$ for belief-state MDP under the policy $\varphi$ and the measure flow $\bmu^*$. Then, for any $f \in C_b(\sX)$, we have
\begin{align}
\mu_{t+1}^*(f) &= \sB(\nu_{t+1,1})(f) \nonumber \\
&= \int_{\sZ \times \sA} \int_{\sZ} z'(f) \eta_t^{\bnu}(dz'|z,a) \nu_t(dz,da) \nonumber \\
&= \int_{\sZ \times \sA} \biggl\{ \int_{\sX} \int_{\sX} f(y) p_t(dy|x,a,\mu_t^*) z(dx) \biggr\} \nu_t(dz,da) \nonumber \\
&= E^{\varphi} \bigl[ l_t(z(t),a(t)) \bigr] \,\, \text{$\biggl($here $l_t(z,a) = \int_{\sX} \int_{\sX} f(y) p_t(dy|x,a,\mu_t^*) z(dx)$$\biggr)$} \nonumber \\
&= E^{\pi^{*}} \biggl[ \int_{\sX} f(y) p_t(dy|x(t),a(t),\mu_t^*) \biggr]. \label{eq:cons}
\end{align}
Since \eqref{eq:cons} is true for all $f \in C_b(\sX)$, we have
\begin{align}
\mu_{t+1}^*(\,\cdot\,) = \int_{\sX \times \sA} p_t(\,\cdot\,|x(t),a(t),\mu_t^*) P^{\pi^{*}}(da(t)|x(t)) \mu_t^*(dx(t)), \nonumber
\end{align}
where $P^{\pi^{*}}(da(t)|x(t))$ denotes the conditional distribution of $a(t)$ given $x(t)$ under $\pi^{*}$ and $(\mu_{\tau}^*)_{0\leq\tau\leq t}$. Hence, $\Lambda(\pi^{*}) = \bmu^*$.

Since $\bnu \in B(\bnu)$, the corresponding Markov policy $\varphi$ satisfies \eqref{eq5} for $\bnu$. Therefore, by Theorem~\ref{theorem1} and the fact that $\nu_{t} = {\cal L}\bigl( z(t),a(t) \bigr)$ for belief-state MDP under the policy $\varphi$ and the measure flow $\bmu^*$, $\varphi$ is optimal for belief-state MDP induced by the measure flow $\bmu^*$ (or, equivalently, $\bnu$). Therefore, $\pi^{*} \in \Psi(\bmu^*)$.\qed
\end{proof}

By Proposition~\ref{prop1}, it suffices to prove that $\Gamma$ has a fixed point in order to establish the existence of a mean-field equilibrium. To prove this, we use Kakutani's fixed point theorem, which is stated below:

\begin{theorem}{\cite[Corollary 17.55]{AlBo06}} Let $K$ be a non-empty compact convex subset of a locally convex Hausdorff space, and let the set-valued mapping $\phi: K \rightarrow 2^K$ have closed graph and non-empty convex values. Then, the set of fixed points of $\phi$ is compact and non-empty.
\end{theorem}

Hence, in order to use Kakutani's fixed point theorem, the set-valued mapping $\Gamma$ should be defined on a convex and compact set. However, the set $\P(\sZ \times \sA)^{T+2}$ in the definition of $\Gamma$ is not compact. To get around that, we will prove that the image of $\P(\sZ \times \sA)^{T+2}$ under $\Gamma$ is in fact a subset of some convex and compact set, and it is sufficient to consider this convex and compact set in the definition of $\Gamma$. To that end, for each $t$, define the set
\begin{align}
\P^t(\sZ) = \biggl\{ \mu \in \P(\sZ): \int_{\sZ} W(z) \mu(dz) \leq \alpha^t M \biggr\}. \nonumber
\end{align}
Since $W$ is a lower semi-continuous moment function, the set $\P^t(\sZ)$ is compact with respect to the weak topology \cite[Proposition E.8, p. 187]{HeLa96}. Let us define
\begin{align}
\P^t(\sZ \times \sA) = \bigl\{ \nu \in \P(\sZ \times \sA): \nu_1 \in \P^t(\sZ) \bigr\}. \nonumber
\end{align}
Since $\sA$ is compact, $\P^t(\sZ \times \sA)$ is tight. Furthermore, $\P^t(\sZ \times \sA)$ is closed with respect to the weak topology since $W$ is lower semi-continuous. Hence, $\P^t(\sZ \times \sA)$ is compact. Let $\Xi = \prod_{t=0}^{T+1} \P^t(\sZ \times \sA)$, which is convex and compact with respect to the product topology.

\begin{proposition}\label{prop2}
We have $\Gamma\bigl(\P(\sZ \times \sA)^{T+2}\bigr) = \bigl\{\bnu' : \bnu' \in \Gamma(\bnu), \text{ } \bnu \in \P(\sZ \times \sA)^{T+2} \bigr\} \subset \Xi$.
\end{proposition}

\begin{proof}
Fix any $\bnu \in \P(\sZ \times \sA)^{T+2}$. It is sufficient to prove that $C(\bnu) \subset \Xi$ as $\Gamma(\bnu) = C(\bnu) \cap B(\bnu)$. Let $\bnu' \in C(\bnu)$. We prove by induction that $\nu'_{t,1} \in \P^t_v(\sZ)$ for all $t$. The claim trivially holds for $t=0$ as $\nu'_{0,1} = \delta_{\mu_0}$. Assume that the claim holds for $t$ and consider $t+1$. We have
\begin{align}
\int_{\sZ} W(y) \nu'_{t+1,1}(dy) &= \int_{\sZ \times \sA} \int_{\sZ} W(y) \eta_{t}^{\bnu}(dy|z,a) \nu_{t}(dz,da) \nonumber \\
&\leq \int_{\sZ} \alpha W(z) \nu_{t,1}(dz) \nonumber \text{ }(\text{by (iv)}) \\
&\leq \alpha^{t+1} M \text{ }(\text{as $\nu_{t,1} \in \P^t_v(\sZ)$}). \nonumber
\end{align}
Hence, $\nu'_{t+1,1} \in \P^{t+1}_v(\sZ)$.\qed
\end{proof}

By Proposition~\ref{prop2}, we can now consider $\Gamma$ as a multi-valued mapping from $\Xi$ into itself. It can be proved that $C(\bnu) \cap B(\bnu) \neq \emptyset$ for any $\bnu \in \Xi$. Indeed, for any $t\geq0$, we define
\begin{align}
\mu_{t+1}(\,\cdot\,) = \int_{\sZ \times \sA} \eta_t^{\bnu}(\,\cdot\,|z,a) \, \nu_t(dx,da).\nonumber 
\end{align}
Moreover, for any $t\geq0$, let $f_t: \sZ \rightarrow \sA$ be the minimizer of the following optimality equation:
\begin{align}
&C_t^{\bnu}(z,f_t(z)) + \int_{\sZ} J_{*,t+1}^{\bnu}(y) \eta_t^{\bnu}(dy|z,f_t(z)) = T_t^{\bnu} J^{\bnu}_{*,t+1}(z). \nonumber
\end{align}
Existence of such an $f_t$ follows from the Measurable Selection Theorem \cite[Section D]{HeLa96} 
since $C_t^{\bnu}$ is continuous in $a$, $\eta_t^{\bnu}$ is weakly continuous in $a$, and $\sA$ is compact. 
If we define $\nu'_t(dz,da) = \mu_t(dz) \, \delta_{f_t(z)}(da)$, then it is straightforward to prove that $\bnu' \in C(\bnu) \cap B(\bnu)$, and thus $C(\bnu) \cap B(\bnu) \neq \emptyset$. 
Moreover, both $C(\bnu)$ and $B(\bnu)$ are convex, and so, their intersection is also convex. $\Xi$ is a convex compact subset of a locally convex topological space $\M(\sZ \times \sA)^{T+2}$, where $\M(\sZ \times \sA)$ denotes the set of all finite signed measures on $\sZ \times \sA$. Hence, in order to deduce the existence of a fixed point of $\Gamma$, we only need to prove that it has a closed graph. Before stating this result, we state the following proposition which
is a key element of the proof.

\begin{proposition}{(\cite[Proposition 4.3]{SaBaRa18})}\label{prop:belief_conv}
Let $\bnu^{(n)} \rightarrow \bnu$ in product topology. Then, for all $t$, $\eta_t^{\bnu^{(n)}}(\,\cdot\,|z_n,a_n)$ weakly converges to $\eta_t^{\bnu}(\,\cdot\,|z,a)$ for all $(z_n,a_n) \rightarrow (z,a) \in \sZ \times \sA$.
\end{proposition}

Using Proposition~\ref{prop:belief_conv}, we can now prove the following result.

\begin{proposition}\label{prop3} The graph of $\Gamma$, i.e., the set
	$$
	\Gr(\Gamma) := \left\{ (\bnu,\bxi) \in \Xi \times \Xi : \bxi \in \Gamma(\bnu)\right\},
	$$
is closed.
\end{proposition}

\begin{proof}
The graph $\Gr(\Gamma)$ of $\Gamma$ is closed if and only if when $(\bnu^{(n)},\bxi^{(n)}) \rightarrow (\bnu,\bxi)$ as $n\rightarrow\infty$ for some $\bigl\{(\bnu^{(n)},\bxi^{(n)})\bigr\} \subset \Xi$, then we must have $\bxi \in \Gamma(\bnu)$. To that end, let $\bigl\{(\bnu^{(n)},\bxi^{(n)})\bigr\} \subset \Gr(\Gamma)$ be such that $(\bnu^{(n)},\bxi^{(n)}) \rightarrow (\bnu,\bxi)$ as $n\rightarrow\infty$ for some $(\bnu,\bxi) \in \Xi \times \Xi$. We prove that $\bxi \in \Gamma(\bnu)$.

Using Proposition~\ref{prop:belief_conv}, we first prove that $\bxi \in C(\bnu)$; that is, for all $t$, we have
\begin{align}
\xi_{t+1,1}(\,\cdot\,) = \int_{\sZ \times \sA} \eta_t^{\bnu}(\,\cdot\,|z,a) \nu_t(dz,da). \nonumber
\end{align}
For all $n$ and $t$, we have
\begin{align}
\xi^{(n)}_{t+1,1}(\,\cdot\,) = \int_{\sZ \times \sA} \eta_t^{\bnu^{(n)}}(\,\cdot\,|z,a) \nu^{(n)}_t(dz,da). \label{eq6}
\end{align}
Since $\bxi^{(n)} \rightarrow \bxi$ in $\Xi$, $\xi^{(n+1)}_{t+1} \rightarrow \xi_{t+1}$ weakly. Let $g \in C_b(\sZ)$. Then, by \cite[Theorem 3.5]{Lan81}, we have
\begin{align}
&\lim_{n\rightarrow\infty} \int_{\sZ \times \sA} \int_{\sZ} g(z') \eta_t^{\bnu^{(n)}}(dz'|z,a) \nu^{(n)}_t(dz,da) =\int_{\sZ \times \sA} \int_{\sZ} g(z') \eta_t^{\bnu}(dz'|z,a)  \nu_t(dx,da) \nonumber
\end{align}
since $\bnu^{(n)}_t \rightarrow \bnu_t$ weakly and $\int_{\sZ} g(y) \eta_t^{\bnu^{(n)}}(\,\cdot\,|z,a)$ converges to $\int_{\sZ} g(y) \eta_t^{\bnu}(\,\cdot\,|z,a)$ continuously\footnote{Suppose $g$, $g_n$ ($n\geq1$) are measurable functions on metric space $\sE$. The sequence $g_n$ is said to converge to $g$ continuously if $\lim_{n\rightarrow\infty}g_n(e_n)=g(e)$ for any $e_n\rightarrow e$ where $e \in \sE$.} (see \cite[Theorem 3.5]{Lan81}). This implies that the measure on the right hand side of \eqref{eq6} converges weakly to $\int_{\sZ \times \sA} \eta_t^{\bnu}(\,\cdot\,|z,a) \nu_t(dz,da)$. Therefore, we have
\begin{align}
\xi_{t+1,1}(\,\cdot\,) = \int_{\sZ \times \sA} \eta_t^{\bnu}(\,\cdot\,|z,a) \nu_t(dz,da), \nonumber
\end{align}
from which we conclude that $\bxi \in C(\bnu)$.

To complete the proof, it suffices to prove that $\bxi \in B(\bnu)$. To that end, for each $n$ and $t$, let us define the following functions
\begin{align}
F^{(n)}_t(z,a) &= C_t^{\bnu^{(n)}}(z,a) +  \int_{\sZ} J^{\bnu^{(n)}}_{*,t+1}(y) \eta_t^{\bnu^{(n)}}(dy|z,a) \nonumber \\
\intertext{and}
F_t(z,a) &= C_t^{\bnu}(z,a) + \int_{\sZ} J^{\bnu}_{*,t+1}(y) \eta_t^{\bnu}(dy|z,a). \nonumber
\end{align}
By definition,
$
J^{\bnu^{(n)}}_{*,t}(z) = \min_{a \in \sA} F^{(n)}_t(z,a) \text{ } \text{ and } \text{ } J^{\bnu}_{*,t}(z) = \min_{a \in \sA} F_t(z,a). \nonumber
$
Define also the following sets
\begin{align}
A_t^{(n)} = \bigl\{ (z,a): F^{(n)}_t(z,a) = J^{\bnu^{(n)}}_{*,t}(z) \bigr\} 
\text{ } \text{and} \text{ }
A_t = \bigl\{ (z,a): F_t(z,a) = J^{\bnu}_{*,t}(z) \bigr\}. \nonumber
\end{align}
Since $\bxi^{(n)} \in B(\bnu^{(n)})$, we have
$
1 = \xi^{(n)}_t\bigl( A_t^{(n)} \bigr), \text{ } \text{for all $n$ and $t$}. \nonumber
$
To prove to $\bxi \in B(\bnu)$, we need to show that
$
1 = \xi_t\bigl( A_t \bigr), \text{ } \text{for all $t$}. \nonumber
$

First note that since both $F^{(n)}_t$ and $J^{\bnu^{(n)}}_{*,t}$ are continuous, $A_t^{(n)}$ is closed. Moreover, $A_t$ is also closed as both $F_t$ and $J^{\bnu}_{*,t}$ are continuous. Using Proposition~\ref{prop:belief_conv}, one can also prove as in \cite[Proposition 3.10]{SaBaRa17}, \cite[Proposition 4.4]{SaBaRa18-r} that $F_t^{(n)}$ converges to $F_t$ continuously and $J^{\bnu^{(n)}}_{*,t}$ converges to $J^{\bnu}_{*,t}$ continuously, as $n\rightarrow\infty$.

For each $M\geq1$, define the closed set $B_t^M = \bigl\{ (z,a): F_t(z,a) \geq J^{\bnu}_{*,t}(z) + \epsilon(M) \bigr\}$, where the sequence $\{\epsilon(M)\}$ is decreasing and $\epsilon(M) \rightarrow 0$ as $M\rightarrow \infty$. Since both $F_t$ and $J^{\bnu}_{*,t}$ are continuous, we can choose $\{\epsilon(M)\}_{M\geq1}$ so that $\xi_t(\partial B_t^M) = 0$ for each $M$. Note that by the monotone convergence theorem, we have
\begin{align}
\xi^{(n)}_t\big(A_t^c \cap A_t^{(n)}\big) = \liminf_{M\to\infty} \xi^{(n)}_t\big(B^M_t \cap A_t^{(n)}). \nonumber
\end{align}
This implies that
\begin{align}
1 &= \limsup_{n\rightarrow\infty} \liminf_{M\rightarrow\infty} \biggl\{ \xi^{(n)}_t\big(A_t \cap A^{(n)}_t\big) + \xi^{(n)}_t\big(B^M_t \cap A_t^{(n)}\big)\biggr\} \nonumber\\
&\leq \liminf_{M\rightarrow\infty} \limsup_{n\rightarrow\infty}  \biggl\{\xi^{(n)}_t\big(A_t \cap A^{(n)}_t\big) + \xi^{(n)}_t\big(B^M_t \cap A_t^{(n)}\big)\biggr\}. \nonumber
\end{align}
\noindent For any fixed $M$, we prove that the limit of the second term in the last expression converges to zero.
To that end, we first note that $\xi^{(n)}_t$ converges weakly to $\xi_t$ as $n\rightarrow\infty$ when both measures are restricted to $B_t^M$, as $B_t^M$ is closed and $\xi_t(\partial B_t^M)=0$ \cite[Theorem 8.2.3]{Bog07}. Furthermore, since $F_t^{(n)}$ converges to $F_t$ continuously and $J^{\bnu^{(n)}}_{*,t}$ converges to $J^{\bnu}_{*,t}$ continuously, $1_{A^{(n)}_t \cap B^M_t}$ converges continuously to $0$, which implies by  \cite[Theorem 3.5]{Lan81} that
\begin{align*}
\limsup_{n\rightarrow\infty} \xi^{(n)}_t\big(B^M_t \cap A^{(n)}_t\big) = 0.
\end{align*}
Therefore, we obtain
\begin{align*}
1 \leq  \limsup_{n\rightarrow\infty} \xi^{(n)}_t\big(A_t \cap A_t^{(n)}\big) \le \limsup_{n\rightarrow\infty} \xi_t^{(n)}(A_t) \leq \xi_t(A_t), \nonumber
\end{align*}
where the last inequality follows from the Portmanteau theorem \cite[Theorem 2.1]{Bil99} and the fact that $A_t$ is closed. Hence, $\xi_t(A_t)=1$. Since $t$ is arbitrary, this is true for all $t$. This means that $\bxi \in B(\bnu)$. Therefore, $\bxi \in \Gamma(\bnu)$.\qed
\end{proof}

As a result of Proposition~\ref{prop3}, we now conclude via Kakutani's fixed point theorem (\cite[Corollary 17.55]{AlBo06}) that $\Gamma$ has a fixed point. Therefore, the pair $(\pi^{*},\bmu^*)$ in Proposition~\ref{prop1} is a mean field equilibrium. This completes the proof of Theorem~\ref{thm:MFE}.

\section{Approximation of Nash Equilibria}\label{sec4}

We are now ready to prove that the policy in the mean-field equilibrium, when applied by every agent, is approximately Nash equilibrium for mean-field games with a sufficiently large number of agents. Let $(\pi^{'*},\bmu^*)$ denote the pair in the mean-field equilibrium. In order to prove the existence of an approximate Nash equilibrium, we need Assumption~\ref{as2} below in addition to Assumption~\ref{as1}.

Our approach can be summarized as follows: (i) First, Assumption~\ref{as2} enables us to define another mean-field equilibrium, in which the policy deterministically and continuously depends on only the observations; (ii) we then construct an equivalent game model whose states are the states of the game model in Section~\ref{sub1sec2} plus the current and past observations; (iii) in this equivalent model, the new mean-field equilibrium policy becomes Markov; (iv) using this Markov structure, we prove that the cost function of a generic agent under any policy in the finite-agent regime, where the rest of the agents adopt mean-field equilibrium policy, converges to the cost function in the infinite-population limit as the number of agents goes to infinity; (v) since the mean-field equilibrium policy is optimal in the infinite-population limit, we establish the existence of an approximate Nash equilibrium via the result in step (iv).

Let $d_{BL}$ denote the bounded Lipschitz metric on $\P(\sS)$, which metrizes the weak topology \cite[Proposition 11.3.2]{Dud89}.

\begin{assumption}\label{as2}
\begin{itemize}
\item [ ]
\item [(a)] $\omega_q(r) \rightarrow 0$ and $\omega_m(r) \rightarrow 0$ as $r\rightarrow0$, where
\begin{align}
\omega_{q}(r) &= \hspace{-10pt} \sup_{(s,u) \in \sS\times\sA} \sup_{\substack{\mu,\nu: \\ d_{BL}(\mu,\nu)\leq r}} \hspace{-10pt} \|q(\,\cdot\,|s,u,\mu) - q(\,\cdot\,|s,u,\nu)\|_{TV} \nonumber \\
\omega_{m}(r) &= \sup_{(s,u) \in \sS\times\sA} \sup_{\substack{\mu,\nu: \\ d_{BL}(\mu,\nu)\leq r}} |m(s,u,\mu) - m(s,u,\nu)|. \nonumber
\end{align}
\item [(b)] For each $t\geq0$, $\pi_t^{'*}: \sG_t \rightarrow \P(\sA)$ is deterministic; that is, $\pi_t^{'*}(\,\cdot\,|g(t)) = \delta_{f_t(g(t))}(\,\cdot\,)$ for some measurable function $f_t:\sG_t\rightarrow \sA$, and weakly continuous.
\end{itemize}
\end{assumption}

In Appendix~\ref{continuity}, we give sufficient conditions for Assumption~\ref{as2}-(b) in terms of the system components.

We now construct another mean-field equilibrium in which the policy deterministically depends on only the observations. For $t$, let $\sY^{t+1} = \prod_{k=0}^t \sY$. Then, for each $t\geq1$, define $\tilde{f}_t:\sY^{t+1}\rightarrow\sA$ as
\begin{align}
&\tilde{f}_t(y(t),\ldots,y(0)) = f_t\bigl(y(t),\ldots,y(0),\tilde{f}_{t-1}(y(t-1),\ldots,y(0)),\ldots,\tilde{f}_0(y(0))\bigr), \nonumber
\end{align}
where $\tilde{f}_0 = f_0$. Let $\pi_t^*(\,\cdot\,|y(t),\ldots,y(0)) = \delta_{\tilde{f}_t(y(t),\ldots,y(0))}(\,\cdot\,)$. Note that $\pi_t^*$ is a weakly continuous stochastic kernel on $\sA$ given $\sY^{t+1}$ under Assumption~\ref{as2}-(b). Moreover, $\pi^*$ and $\pi^{'*}$ are equivalent because, for all $t$, we have
\begin{align}
P^{\pi^{'*}}\bigl(a(t) \in \,\cdot\,|g(t)\bigr) &= P^{\pi^{'*}}\bigl(a(t) \in \,\cdot\,|y(t),\ldots,y(0)\bigr) \nonumber \\
&= P^{\pi^*}\bigl(a(t) \in \,\cdot\,|y(t),\ldots,y(0)\bigr). \nonumber
\end{align}
Hence, $(\pi^*,\bmu^*)$ is also a mean-field equilibrium. In the sequel, we use $(\pi^*,\bmu^*)$ to prove the approximation result. The reason for passing from $f_t$ to $\tilde{f}_t$ is that the latter policy becomes Markov in the equivalent game model that will be introduced in the proof of Theorem~\ref{appr-thm}. Then, we can prove the existence of an approximate Nash equilibrium by adapting the proof techniques and results in \cite{SaBaRa17,SaBaRa18-r} to the game models with expanding state spaces and non-homogeneous system components.

The following theorem is the main result of this section, which states that the policy ${\boldsymbol \pi}^{(N,*)} = (\pi^*,\ldots,\pi^*)$, where $\pi^*$ is repeated $N$ times, is an $\varepsilon$-Nash equilibrium for sufficiently large $N$. Its proof appears in the next section.

\begin{theorem}\label{appr-thm}
For any $\varepsilon>0$, there exists $N(\varepsilon)$ such that for $N\geq N(\varepsilon)$, the policy ${\boldsymbol \pi}^{(N,*)}$ is an $\varepsilon$-Nash equilibrium for the game with $N$ agents that is introduced in Section~\ref{sub1sec2}. Since the original $N$-agent game model is equivalent to the one in Section~\ref{sub1sec2} by Proposition~\ref{app-prop1}, the policy ${\boldsymbol \pi}^{(N,*)}$ is also an $\varepsilon$-Nash equilibrium for the original game with $N$ agents.
\end{theorem}

\begin{remark}
Note that to obtain an explicit relation between $\varepsilon$ and $N(\varepsilon)$, one needs to establish that the optimal policy $\pi^*$ in mean-field equilibrium is Lipschitz continuous. In the \emph{fully-observed continuous-time} setup, this is in general established easily due to very restrictive structural assumptions on the system components. In a recent monograph \cite{DeCa18}, Lipschitz continuity of the optimal policy in mean-field equilibrium was established in Lemma 3.3 using regularity properties of system components. However, in our setup, in order to establish this, we need Lipschitz continuity, strong convexity, and differentiability conditions on one-stage cost functions $\{C_t\}$ and transition probabilities $\{\eta_t\}$ of the fully-observed reduction. However, establishing Lipschitzness of the transition probabilities $\{\eta_t\}$ is in general prohibitive. Indeed, even weak continuity of the transition probabilities $\{\eta_t\}$, which is a much weaker condition than Lipschitz continuity, has been established relatively recently in \cite{FeKaZg16}. Moreover, it was discussed in that paper that even if very restrictive conditions are imposed on the system components, it is not possible to extend weak continuity of the transition probability to setwise continuity, which is also a very weak condition that is used in the stochastic control literature. Therefore, establishing Lipschitz continuity of the transition probabilities $\{\eta_t\}$ is in general prohibitive. This would also be the case for the partially-observed continuous-time setup, since the above-mentioned result pertains to the fully-observed case.
\end{remark}

\begin{remark}
In the mean-field games literature, uniqueness of the mean-field equilibrium can be established using a monotonicity condition as introduced by Lasry and Lions in \cite{LaLi07} (see also \cite{CaLa2015}). However, in addition to the monotonicity condition, we should also have the following conditions in order to have uniqueness (see, e.g., \cite[Assumption U]{CaLa2015}):
\begin{itemize}
\item[a)] The cost function should be in additive form.
\item[b)] The one-stage cost function can be additively decomposed into two functions, where the first function is a function of the state and the mean-field term, and the second function is a function of the state and the action. 
\item[c)] The dynamics of a generic agent should be independent of the mean-field term.
\item[d)] For any state-measure flow, there exists a unique optimal policy.
\end{itemize}   

Under these conditions, one can prove that if $(\pi^{\bmu},\bmu)$ and $(\pi^{\bnu},\bnu)$ are two mean-field equilibria, then 
\begin{align}
J_{\bmu}(\pi^{\bmu}) + J_{\bnu}(\pi^{\bnu}) \geq J_{\bmu}(\pi^{\bnu}) + J_{\bnu}(\pi^{\bmu}) \label{eq1}
\end{align}   
in the equivalent game model. This implies that $J_{\bmu}(\pi^{\bmu}) = J_{\bmu}(\pi^{\bnu})$ and $J_{\bnu}(\pi^{\bmu}) = J_{\bnu}(\pi^{\bmu})$. Then, conditions c) and d) ensure that these mean-field equlibria must be the same, which implies uniqueness. However, note that to have inequality (\ref{eq1}), conditions a), b), and c) must hold. Indeed, to state the monotonicity condition, we should have condition b).

In our case, the cost function in the equivalent game model is in additive form, and thus we do have condition a). Moreover, we can assume the decomposition in condition b). However, if we assume that transition probabilities $\{p_t\}$ are independent of the mean-field term, then it implies that the transition probability $q$ and the one-stage cost function $m$ of the original game model are independent of the mean-field term since 
\begin{align}
p_t\bigl(B \times D \big| x(t),a(t),{\color{red}\mu_t}\bigr) = q(B|s(t),a(t),{\color{red}\mu_{t,1}}) \otimes \delta_{m(t) + \beta^t m(s(t),a(t),{\color{red}\mu_{t,1}})}(D). \nonumber
\end{align}
But this is merely a risk-sensitive stochastic control setup. 

Conversely, if we consider the original game model instead of the equivalent one, then, in this case, the cost function is not in additive form and thus, we cannot achieve inequality (\ref{eq1}) because we cannot have conditions a) and b), which are needed along with the monotonicity condition to have unique mean-field equilibrium. 
\end{remark}


\section{Proof of Theorem~\ref{appr-thm}}\label{sec4-1}

For the game model introduced in Section~\ref{sub1sec2}, the policy $\pi^*$ in the mean-field equilibrium is not necessarily Markov, and so, the joint process of the state, observation, and mean-field term does not have the Markov property as well. To prove Theorem~\ref{appr-thm}, we will first introduce another equivalent game model whose states are the state of the original game model\footnote{When we say original game model in this section, it means the game model introduced in Section~\ref{sub1sec2} in place of the risk-sensitive game model.} plus the current and past observations. In this new model, the mean-field equilibrium policy automatically becomes Markov. 

\def\sS{{\mathsf B}}

In the infinite-population limit, this new mean-field game model is specified by
\begin{align}
\biggl( \{\sS_t\}_{t=0}^{T+1}, \sA, \{P_t\}_{t=0}^{T+1}, \{\C_t\}_{t=0}^{T+1}, \lambda_0 \biggr), \nonumber
\end{align}
where, for each $t$,
$
\sS_t = \sX \times \underbrace{\sY \times \ldots \times \sY}_{\text{$t+1$-times}} \nonumber
$
and $\sA$ are the Polish state and action spaces at time $t$, respectively. The stochastic kernel $P_t : \sS_t \times \sA \times \P(\sS_t) \to \P(\sS_{t+1})$ is defined as:
\begin{align}
&P_t\bigl(B_{t+1} \times D_{t+1} \times \ldots \times D_0 \big| b(t),a(t),\Delta_t\bigr) \nonumber \\
&= \int_{B_{t+1}} r(D_{t+1}|x(t+1)) \prod_{k=0}^t 1_{D_k}(y(k)) p_t(dx(t+1)|x(t),a(t),\Delta_{t,1}) , \nonumber
\end{align}
where $B_{t+1} \in \B(\sX)$, $D_k \in \B(\sY)$ ($k=0,\ldots,t+1$), $b(t) = (x(t),y(t),y(t-1),\ldots,y(0))$, and $\Delta_{t,1}$ is the marginal of $\Delta_t$ on $\sX$. Indeed, $P_t$ is the controlled transition probability of next state-observation pair, current observation, and past observations, i.e.,
$\bigl(x(t+1),y(t+1),y(t),\ldots,y(0)\bigr),$
given the current state-observation pair and past observations, i.e.,
$\bigl(x(t),y(t),y(t-1),\ldots,y(0)\bigr),$
in the original mean-field game. For each $t$, the one-stage cost function $\C_t: \sS_t \times \sA \times \P(\sS_t) \rightarrow [0,\infty)$ (do not confuse this with $C_t$ in Section~\ref{main-proof}) is defined as:
\begin{align}
\C_t(b(t),a(t),\Delta_t) = c_t(x(t),a(t),\Delta_{t,1}). \nonumber
\end{align}
Finally, the initial measure $\lambda_0$ is given by $\lambda_0(db) = r(dy|x) \mu_0(dx)$, where $b = (x,y)$. Suppose that Assumption~\ref{as1} and Assumption~\ref{as2} hold.
Then, for each $t$, the following are satisfied:
\begin{itemize}
\item [(I)] The one-stage cost function $\C_t$ is bounded and continuous.
\item [(II)] The stochastic kernel $P_t$ is weakly continuous.
\end{itemize}
It is straightforward to prove that (I) and (II) hold since $c_t$ is continuous, $p_t$ is weakly continuous, and $r$ is continuous in total variation norm. Recall the set of policies $\tilde{\Pi}$ in the original mean-field game which only use the observations; that is, $\pi \in \tilde{\Pi}$ if $\pi_t:\sY^{t+1} \rightarrow \P(\sA)$ for each $t\geq0$. Note that $\tilde{\Pi}$ is a subset of the set of Markov policies in the new model. For any measure flow ${\boldsymbol \Delta} = (\Delta_t)_{t\geq0}$, where $\Delta_t \in \P(\sS_t)$, we denote by $\hat{J}_{{\boldsymbol \Delta}}(\pi)$ the finite-horizon risk-neutral total cost of the policy $\pi \in \tilde{\Pi}$ in this new mean-field game model.

We also define the corresponding $N$ agent game as follows. We have the Polish state spaces $\{\sS_t\}_{t=0}^{T+1}$ and action space $\sA$. For every $t$ and every $i \in \{1,2,\ldots,N\}$, let $b^N_i(t) \in \sS_t$ and $a^N_i(t) \in \sA$ denote the state and the action of Agent~$i$ at time $t$, and let
\begin{align}
\Delta_t^{(N)}(\,\cdot\,) = \frac{1}{N} \sum_{i=1}^N \delta_{b_i^N(t)}(\,\cdot\,) \in \P(\sS_t) \nonumber
\end{align}
denote the empirical distribution of the state configuration at time $t$. The initial states $b^N_i(0)$ are independent and identically distributed according to $\lambda_0$, and, for each $t$, the next-state configuration $(b^N_1(t+1),\ldots,b^N_N(t+1))$ is generated according to the probability laws
\begin{align}
&\prod^N_{i=1} P_{t}\big(db^N_i(t+1)\big|b^N_i(t),a^N_i(t),\Delta^{(N)}_t\big). \nonumber 
\end{align}
Recall that $\tilde{\Pi}_i$ denotes the set of policies that only use local observations for Agent $i$ in the original game. Note that policies in $\tilde{\Pi}_i$ are Markov for the new model since they partly use the state information. We let $\tilde{\Pi}_i^c$ denote the set of all policies in $\tilde{\Pi}_i$ for Agent~$i$ that are weakly continuous; that is, $\pi=\{\pi_t\}\in \tilde{\Pi}_i^c$ if for all $t\geq0$, $\pi_t: \sY^{t+1} \rightarrow \P(\sA)$ is continuous when $\P(\sA)$ is endowed with the weak topology. For Agent~$i$, the finite-horizon risk-neutral total cost under the initial distribution $\lambda_0$ and $N$-tuple of policies ${\boldsymbol \pi}^{(N)} \in \tilde{{\bf \Pi}}^{(N)}$ is denoted by $\hat{J}_i^{(N)}({\boldsymbol \pi}^{(N)})$.

The following proposition makes the connection between this new model and the original model.

\begin{proposition}\label{oapp-prop1}
For any $N\geq1$, ${\boldsymbol \pi}^{(N)} \in \tilde{{\bf \Pi}}^{(N)}$, and $i=1,\ldots,N$, we have $\hat{J}_i({\boldsymbol \pi}^{(N)}) = J_i({\boldsymbol \pi}^{(N)})$. Similarly, for any $\pi \in \tilde{\Pi}$ and measure flow ${\boldsymbol \Delta}$, we have $\hat{J}_{\boldsymbol \Delta}(\pi) = J_{\bmu}(\pi)$ where $\bmu = (\Delta_{t,1})_{t\geq0}$.
\end{proposition}

\begin{proof}
The result can easily be proved as in \cite[Proposition 5.1]{SaBaRa18}, and thus we do not include the details.\qed
\end{proof}

By Proposition~\ref{oapp-prop1}, in the remainder of this section we consider the new game model in place of the one introduced in Section~\ref{sub1sec2}. Define the measure flow ${\boldsymbol \Delta} = (\Delta_t)_{t\geq0}$ as follows:
$$\Delta_t = {\cal L}(x(t),y(t),\ldots,y(0)),$$
where ${\cal L}(x(t),y(t),\ldots,y(0))$ denotes the probability law of
$(x(t),y(t),\ldots,y(0))$ in the original mean-field game under the policy $\pi^*$ in the mean-field equilibrium.
For each $t\geq0$, define the stochastic kernel $P_t^{\pi^*}(\,\cdot\,|b,\Delta)$ on $\sS_{t+1}$ given $\sS_{t} \times \P(\sS_{t})$ as
\begin{align}
P_t^{\pi^*}(\,\cdot\,|b,\Delta) = \int_{\sA} P_t(\,\cdot\,|b,a,\Delta) \pi_t^*(da|b). \nonumber
\end{align}
Since $\pi_t^*$ is weakly continuous, $P_t^{\pi^*}(\,\cdot\,|b,\Delta)$ is also weakly continuous in $(b,\Delta)$. In the sequel, to ease the notation, we will also write $P_t^{\pi^*}(\,\cdot\,|b,\Delta)$ as $P_{t,\Delta}^{\pi^*}(\,\cdot\,|b)$.

\begin{lemma}\label{app-lemma1}
Measure flow ${\boldsymbol \Delta}$ satisfies
\begin{align}
\Delta_{t+1}(\,\cdot\,) &= \int_{\sS_t} P_{t}^{\pi^*}(\,\cdot\,|b,\Delta_t) \Delta_t(db) \nonumber \\
&= \Delta_t P_{t,\Delta_t}^{\pi^*}(\,\cdot\,). \nonumber
\end{align}
\end{lemma}

\begin{proof}
The result can easily be proved as in \cite[Lemma 5.1]{SaBaRa18}, and thus we do not include the details.\qed
\end{proof}

For each $N\geq1$, let $\bigl\{b_i^{N}(t)\bigr\}_{1\leq i\leq N}$ denote the states of agents at time $t$ in the $N$-agent new game model under the policy ${\boldsymbol \pi}^{(N,*)} = \{\pi^*,\pi^*,\ldots,\pi^*\}$. Define the empirical distribution
\begin{align}
\Delta_t^{(N)}(\,\cdot\,) = \frac{1}{N} \sum_{i=1}^N \delta_{b_i^{N}(t)}(\,\cdot\,). \nonumber
\end{align}

\begin{proposition}\label{prop5}
For all $t\geq0$, we have
$
{\cal L}(\Delta_t^{(N)}) \rightarrow \delta_{\Delta_t} \nonumber
$
weakly in $\P(\P(\sS_t))$, as $N\rightarrow\infty$.
\end{proposition}

\begin{proof}
Weak topology on $\P(\sS_t)$ can be metrized using the following metric:
\begin{align}
\rho(\mu,\nu) = \sum_{m=1}^{\infty} 2^{-(m+1)} | \mu(f_m) - \nu(f_m) |, \nonumber
\end{align}
where $\{f_m\}_{m\geq1}$ is a sequence of real continuous and bounded functions on $\sS_t$ such that $\|f_m\| \leq 1$ for all $m\geq1$ (see \cite[Theorem 6.6, p. 47]{Par67}). Define the Wasserstein distance of order 1 on the set of probability measures $\P(\P(\sS_t))$ as follows (see \cite[Definition 6.1]{Vil09}):
\begin{align}
W_1(\Phi,\Psi) = \inf \bigl\{ E[\rho(X,Y)]: {\cal L}(X) = \Phi \text{ and } {\cal L}(Y) = \Psi \bigr\}. \nonumber
\end{align}
Note that since $\delta_{\Delta_t}$ is a Dirac measure, we have
\begin{align}
W_1({\cal L}(\Delta_t^{(N)}),\delta_{\Delta_t}) &= \bigl\{ E[\rho(X,Y)]: {\cal L}(X) = {\cal L}(\Delta_t^{(N)}) \text{ and } {\cal L}(Y) = \delta_{\Delta_t} \bigr\} \nonumber \\
&= E\biggl[ \sum_{m=1}^{\infty} 2^{-(m+1)} | \Delta_t^{(N)}(f_m) - \Delta_t(f_m) | \biggr]. \nonumber
\end{align}
Since convergence in $W_1$ distance implies weak convergence (see \cite[Theorem 6.9]{Vil09}), it suffices to prove that
\begin{align}
\lim_{N\rightarrow\infty} E\bigl[|\Delta_t^{(N)}(f) - \Delta_t(f)|\bigr] = 0 \nonumber
\end{align}
for any $f \in C_b(\sS_t)$ and for all $t$. We prove this by induction on $t$.

As $\{b_i^N(0)\}_{1\leq i\leq N}$ are i.i.d. with common distribution $\Delta_0$, the claim is true for $t=0$. We suppose that the claim holds for $t$ and consider $t+1$. Fix any $g \in C_b(\sS_{t+1})$. Then, we have
\begin{align}
&|\Delta_{t+1}^{(N)}(g) - \Delta_{t+1}(g)| \nonumber \\
&\phantom{xxx}\leq |\Delta_{t+1}^{(N)}(g) - \Delta_{t}^{(N)} P^{\pi^*}_{t,\Delta_t^{(N)}}(g)| + |\Delta_t^{(N)} P^{\pi^*}_{t,\Delta_t^{(N)}}(g) - \Delta_t P^{\pi^*}_{t,\Delta_t}(g) |. \label{eq8}
\end{align}
We first prove that the expectation of the second term on the right-hand side (RHS) of \eqref{eq8} converges to $0$ as $N\rightarrow\infty$. To that end, define $F: \P(\sS_{t}) \rightarrow \R$ as
\begin{align}
F(\Delta) = \Delta P^{\pi^*}_{t,\Delta}(g) = \int_{\sS_t} \int_{\sS_{t+1}} g(b') P^{\pi^*}_t(db'|b,\Delta) \Delta(db). \nonumber
\end{align}

One can prove that $F \in C_b(\P(\sS_t))$. Indeed, suppose that $\Delta_n$ converges to $\Delta$. Let us define
\begin{align}
l_n(b) &= \int_{\sS_{t+1}} g(b') P^{\pi^*}_t(db'|b,\Delta_n) 
\text{ } \text{and} \text{ }
l(b) = \int_{\sS_{t+1}} g(b') P^{\pi^*}_t(db'|b,\Delta). \nonumber
\end{align}
Since $P^{\pi^*}_t$ is weakly continuous, one can prove that $l_n$ converges to $l$ continuously. By \cite[Theorem 3.5]{Lan81}, we have $F(\Delta_n) \rightarrow F(\Delta)$, and so, $F \in C_b(\P(\sS_t))$. This implies that the expectation of the second term on the RHS of \eqref{eq8} converges to zero as ${\cal L}(\Delta_t^{(N)}) \rightarrow \delta_{\Delta_t}$ weakly, by the induction hypothesis.

Now, let us write the expectation of the first term on the RHS of \eqref{eq8} as
\begin{align}
E\biggl[ E\biggl[ |\Delta_{t+1}^{(N)}(g) - \Delta_t^{(N)} P^{\pi^*}_{t,\Delta_t^{(N)}}(g)| \biggr| b_1^N(t),\ldots,b_N^N(t) \biggr] \biggr]. \nonumber
\end{align}
Then, by \cite[Lemma A.2]{BuMa14}, we have
\begin{align}
E\biggl[ |\Delta_{t+1}^{(N)}(g) - \Delta_t^{(N)} P^{\pi^*}_{t,\Delta_t^{(N)}}(g)| \biggr| b_1^N(t),\ldots,b_N^N(t) \biggr] \leq 2 \frac{\|g\|}{\sqrt{N}}. \nonumber
\end{align}
Therefore, the expectation of the first term on the RHS of \eqref{eq8} also converges to zero as $N\rightarrow\infty$. Since $g$ was arbitrary, this completes the proof.\qed
\end{proof}

The implication of Proposition~\ref{prop5} is the key to prove the main theorem. It basically says that, in the infinite-population limit, the empirical distribution of the states under the mean-field policy converges to the deterministic measure flow ${\boldsymbol \Delta}$ (i.e., the principle of law of large numbers). This result leads to the following important proposition.

\begin{proposition}\label{prop6}
We have
\begin{align}
\lim_{N\rightarrow\infty} \hat{J}_1^{(N)}({\boldsymbol \pi}^{(N,*)}) = \hat{J}_{{\boldsymbol \Delta}}(\pi^*) = \inf_{\pi' \in \Pi} \hat{J}_{{\boldsymbol \Delta}}(\pi'). \nonumber
\end{align}
\end{proposition}

\begin{proof}
As the transition probabilities $P_t(\,\cdot\,|d,a,\Delta)$ are continuous in $\Delta$, the dynamics of the state of a generic agent in the finite-agent game with sufficiently many agents and the dynamics of the state in the mean-field game under policies $\bpi^{(N,*)} = (\pi^*,\ldots,\pi^*)$ and $\pi^*$, respectively, should therefore be close. Hence, the distributions of the states in these games should also be close, from which we obtain the proposition. The precise mathematical proof is given below.

For each $t\geq0$, let us define
\begin{align}
\C_{\pi_t^*}(b,\Delta) = \int_{\sA} \C_t(b,a,\Delta) \pi_t^*(da|b). \nonumber
\end{align}
Note that random elements $\bigl(b_1^N(t),\ldots,b_N^N(t),\Delta_t^{(N)}\bigr)$ are exchangeable; that is, for any permutation $\sigma$ of $\{1,\ldots,N\}$, we have
\begin{align}
{\cal L}\bigl(b_1^N(t),\ldots,b_N^N(t),\Delta_t^{(N)}\bigr) = {\cal L}\bigl(b_{\sigma(1)}^N(t),\ldots,b_{\sigma(N)}^N(t),\Delta_t^{(N)}\bigr). \nonumber
\end{align}
Hence, the cost function at time $t$ can be written as
\begin{align}
E\bigl[ \C_t(b_1^N(t),a_1^N(t),\Delta_t^{(N)}) \bigr] &= \frac{1}{N} \sum_{i=1}^N E\bigl[ \C_t(b_i^N(t),a_i^N(t),\Delta_t^{(N)}) \bigr] \nonumber \\
&= E\bigl[ \Delta_t^{(N)}\bigl(\C_{\pi_t^*}(b,\Delta_t^{(N)})\bigr) \bigr]. \nonumber
\end{align}
Define $F: \P(\sS_t) \rightarrow \R$ as
\begin{align}
F(\Delta) = \int_{\sS_t} \C_{\pi_t^*}(b,\Delta) \Delta(db). \nonumber
\end{align}
One can show that $F \in C_b(\P(\sS_t))$ as $\pi_t^*$ is weakly continuous. Hence, by Proposition~\ref{prop5}, we obtain
\begin{align}
\lim_{N\rightarrow\infty} E\bigl[ \C_t(b_1^N(t),a_1^N(t),\Delta_t^{(N)}) \bigr] &= \lim_{N\rightarrow\infty} E\bigl[ \Delta_t^{(N)}\bigl(\C_{\pi_t^*}(b,\Delta_t^{(N)})\bigr) \bigr] \nonumber \\
&= \lim_{N\rightarrow\infty} E[F(\Delta_t^{(N)})] \nonumber \\
&= F(\Delta_t) \nonumber \\
&= \Delta_t(\C_{\pi_t^*}(\,\cdot\,,\Delta_t)). \label{eq9}
\end{align}
Note that by Lemma~\ref{app-lemma1}, the cost in the mean-field game can be written as
\begin{align}
\hat{J}_{{\boldsymbol \Delta}}(\pi^*) = \sum_{t=0}^{T+1} \Delta_t(\C_{\pi_t^*}(\,\cdot\,,\Delta_t)). \nonumber
\end{align}
Therefore, by \eqref{eq9} and the dominated convergence theorem, we obtain
\begin{align}
\lim_{N\rightarrow\infty} \hat{J}_1^{(N)}({\boldsymbol \pi}^{(N,*)}) = \hat{J}_{{\boldsymbol \Delta}}(\pi^*), \nonumber
\end{align}
which completes the proof.\qed
\end{proof}

To obtain the approximation result, we should show that if the policy of some agent deviates from the mean-field equilibrium policy, then the corresponding cost of this agent should be close to the cost in the mean-field limit as in Proposition~\ref{prop6}, for $N$ sufficiently large. Since the transition probabilities and the one-stage cost functions are identical for all agents in the game model, it is sufficient to change the policy of Agent~$1$ for each $N$. To that end, let $\{\tpi^{(N)}\}_{N\geq1} \subset \tilde{\Pi}_1^c$ be an arbitrary sequence of policies for Agent~$1$; that is, for each $N\geq1$ and $t\geq0$, $\tpi_t^{(N)}: \sY^{t+1} \rightarrow \P(\sA)$ is weakly continuous. For each $N\geq1$, let $\bigl\{\tb_i^N(t)\bigr\}_{1\leq i \leq N}$ be the collection of states in the $N$-person game under the policy $\tilde{{\boldsymbol \pi}}^{(N)} = \{\tpi^{(N)},\pi^*,\ldots,\pi^*\}$. Define
\begin{align}
\tilde{\Delta}_t^{(N)}(\,\cdot\,) = \frac{1}{N} \sum_{i=1}^N \delta_{\tb_i^{(N)}(t)}(\,\cdot\,). \nonumber
\end{align}
The following result says that the asymptotic behaviour of the empirical distribution of the states at each time $t$ is insensitive to local deviations from the mean-field equilibrium policy.

\begin{proposition}\label{prop8}
For all $t\geq0$, we have
$
{\cal L}(\tilde{\Delta}_t^{(N)}) \rightarrow \delta_{\Delta_t} \nonumber
$
weakly $\P(\P(\sS_t))$, as $N \rightarrow \infty$.
\end{proposition}

\begin{proof}
The proof can be done by slightly modifying the proof of Proposition~\ref{prop5}, and therefore will not be included here.\qed
\end{proof}

For each $N\geq1$, let $\{\hb^N(t)\}_{t\geq0}$ denote the state trajectory of the generic agent in the mean-field game (i.e., infinite-population limit) under policy $\tpi^{(N)}$; that is, $\hb^N(t)$ evolves as follows:
\begin{align}
\hb^N(0) \sim \lambda_0 \text{ and } \hb^N(t+1) \sim P^{\tpi^{(N)}}_{t,\Delta_t}(\,\cdot\,|\hb^N(t)). \nonumber
\end{align}
The cost function of this mean-field game is given by
\begin{align}
\hat{J}_{{\boldsymbol \Delta}}(\tpi^{(N)}) = \sum_{t=0}^{T+1} E\bigl[ C_t(\hb^N(t),\hat{a}^N(t),\Delta_t)\bigr], \label{nneq1}
\end{align}
where the actions at each time $t\geq0$ is generated according to the probability law
\begin{align}
\tilde{\pi}^{(N)}_t(d\hat{a}^N(t)|\hb^N(t)) = \tilde{\pi}^{(N)}_t(d\hat{a}^N(t)|\hat{y}^N(t),\ldots,\hat{y}^N(0)).\nonumber
\end{align}

The following result is a bit technical but very important for proving the main result. Its proof is quite long and complicated, and thus can be found in Appendix~\ref{app3}. 

\begin{proposition}\label{prop9}
For any $t\geq0$, we have
\begin{align}
\lim_{N\rightarrow\infty} \bigl| {\cal L}(\tb_1^N(t))(g_N) - {\cal L}(\hb^N(t))(g_N) \bigr| = 0 \nonumber
\end{align}
for any sequence $\{g_N\} \subset C_b(\sS_t)$ such that $\sup_{N\geq1}\|g_N\|<\infty$ and $\omega_g(r) \rightarrow 0$ as $r \rightarrow 0$, where
\begin{align}
\omega_g(r) = \sup_{\substack{s \in {\mathsf S} \\ y^t \in \sY^t}} \sup_{N\geq1} \sup_{\substack{m,m' \\ |m - m'| \leq r}} |g_N(s,m,y^t) - g_N(s,m',y^t)|. \nonumber
\end{align}
\end{proposition}

Using Proposition~\ref{prop9}, we now prove the following result.

\begin{theorem}\label{theorem3}
Let $\{\tpi^{(N)}\}_{N\geq1} \subset \tilde{\Pi}_1^c$ be an arbitrary sequence of policies for Agent~$1$. Then, we have
\begin{align}
\lim_{N \rightarrow \infty} \bigl| \hat{J}_1^{(N)}(\tpi^{(N)},\pi^*,\ldots,\pi^*) - \hat{J}_{{\boldsymbol \Delta}}(\tpi^{(N)}) \bigr| = 0, \nonumber
\end{align}
where $\hat{J}_{{\boldsymbol \Delta}}(\tpi^{(N)})$ is given in \eqref{nneq1}.
\end{theorem}

\begin{proof}
Since $\C_t = 0$ for $t \leq T$, we set $t=T+1$. We have
\begin{align}
&\bigl| \hat{J}_1^{(N)}(\tpi^{(N)},\pi^*,\ldots,\pi^*) - \hat{J}_{{\boldsymbol \Delta}}(\tpi^{(N)}) \bigr| = \bigl| E\bigl[ \C_t(\tb_1^N(t)) \bigr] - E\bigl[ \C_t(\hb_1^N(t)) \bigr] \bigr|. \nonumber
\end{align}
Note that $\C_t(b) = \C_t((s,m,y_0,\ldots,y_t)) = e^{\lambda m}$, where $m \in [0,L]$, is Lipschitz. Therefore, the term in the above equation converges to zero by Proposition~\ref{prop9}.\qed
\end{proof}

As a corollary of Proposition~\ref{prop6} and Theorem~\ref{theorem3}, we obtain the following result.

\begin{corollary}\label{cor1}
We have
\begin{align}
\lim_{N \rightarrow \infty} \hat{J}_1^{(N)}(\tpi^{(N)},\pi^*,\ldots,\pi^*)
&\geq \inf_{\pi' \in \tilde{\Pi}} \hat{J}_{{\boldsymbol \Delta}}(\pi') = \hat{J}_{{\boldsymbol \Delta}}(\pi^*) \nonumber \\
&= \lim_{N \rightarrow \infty} \hat{J}_1^{(N)}(\pi^*,\pi^*,\ldots,\pi^*), \nonumber
\end{align}
where $\{\tpi^{(N)}\}_{N\geq1} \subset \tilde{\Pi}_1^c$ is an arbitrary sequence of policies for Agent~$1$.
\end{corollary}

Now, we are ready to prove the main result of this section.

\begin{proof}{(Proof of Theorem~\ref{appr-thm})}
One can prove that for any policy ${\boldsymbol \pi}^{(N)} \in \tilde{{\bf \Pi}}^{(N)}$, we have
\begin{align}
\inf_{\pi^i \in \tilde{\Pi}_i} \hat{J}_i^{(N)}({\boldsymbol \pi}^{(N)}_{-i},\pi^i) = \inf_{\pi^i \in \tilde{\Pi}_i^c} \hat{J}_i^{(N)}({\boldsymbol \pi}^{(N)}_{-i},\pi^i) \nonumber
\end{align}
for each $i=1,\ldots,N$ (see the proof of \cite[Theorem 2.3]{SaBaRa17}). Hence, it is sufficient to consider weakly continuous policies in ${\bf \Pi}^{(N)}$ to establish the existence of $\varepsilon$-Nash equilibrium in the new model.

We prove that, for sufficiently large $N$, we have
\begin{align}
\hat{J}_i^{(N)}({\boldsymbol \pi}^{(N,*)}) &\leq \inf_{\pi^i \in \tilde{\Pi}_i^c} \hat{J}_i^{(N)}({\boldsymbol \pi}^{(N,*)}_{-i},\pi^i) + \varepsilon \label{eq13}
\end{align}
for each $i=1,\ldots,N$. As indicated earlier, since the transition probabilities and the one-stage cost functions are the same for all agents in the new game, it is sufficient to prove \eqref{eq13} for Agent~$1$ only. Given $\epsilon > 0$, for each $N\geq1$, let $\tpi^{(N)} \in \tilde{\Pi}_1^c$ be such that
\begin{align}
\hat{J}_1^{(N)} (\tpi^{(N)},\pi^*,\ldots,\pi^*) < \inf_{\pi' \in \tilde{\Pi}_1^c} \hat{J}_1^{(N)} (\pi',\pi^*,\ldots,\pi^*) + \frac{\varepsilon}{3}. \nonumber
\end{align}
Then, by Corollary~\ref{cor1}, we have
\begin{align}
\lim_{N\rightarrow\infty} \hat{J}_1^{(N)} (\tpi^{(N)},\pi^*,\ldots,\pi^*) &= \lim_{N\rightarrow\infty} \hat{J}_{{\boldsymbol \Delta}}(\tpi^{(N)}) \nonumber \\
&\geq \inf_{\pi'} \hat{J}_{{\boldsymbol \Delta}}(\pi') \nonumber \\
&= \hat{J}_{{\boldsymbol \Delta}}(\pi^*) \nonumber \\
&= \lim_{N\rightarrow\infty} \hat{J}_1^{(N)} (\pi^*,\pi^*,\ldots,\pi^*). \nonumber
\end{align}
Therefore, there exists $N(\varepsilon)$ such that for $N\geq N(\varepsilon)$, we have
\begin{align}
\inf_{\pi' \in \tilde{\Pi}_1^c} \hat{J}_1^{(N)} (\pi',\pi^*,\ldots,\pi^*) + \varepsilon &> \hat{J}_1^{(N)} (\tpi^{(N)},\pi^*,\ldots,\pi^*) + \frac{2\varepsilon}{3} \nonumber \\
&\geq \hat{J}_{{\boldsymbol \Delta}}(\pi^*) + \frac{\varepsilon}{3} \nonumber \\
&\geq \hat{J}_1^{(N)} (\pi^*,\pi^*,\ldots,\pi^*). \nonumber
\end{align}
The result then follows from Proposition~\ref{oapp-prop1}.\qed
\end{proof}

\section{Infinite Horizon Cost Function}\label{infinite}

In this section, we extend Theorem~\ref{appr-thm} to games with infinite-horizon risk-sensitive cost functions; that is, a generic agent's infinite-horizon risk-sensitive cost under the initial distribution $\kappa_0$ and the $N$-tuple of infinite-horizon policies ${\boldsymbol \pi}^{(N,\infty)}=(\pi^{(1,\infty)},\ldots,\pi^{(N,\infty)}) \in {\bf \Pi}^{(N)}$ is given by
\begin{align}
W_i^{(N,\infty)}({\boldsymbol \pi}^{(N,\infty)}) &=  E^{{\boldsymbol \pi}^{(N,\infty)}}\biggl[ e^{\lambda\sum_{t=0}^{\infty}\beta^{t}m(s_{i}^N(t),u_{i}^N(t),d^{(N)}_t)}\biggr],  \nonumber
\end{align}
where, for each Agent~$j$, $\pi^{(j,\infty)} = \{\pi^{(j,\infty)}_0,\pi^{(j,\infty)}_1,\ldots\}$ (i.e., infinitely many stochastic kernels). Note that, by \cite[Lemma 4.3]{SaBaRa18-r}, any infinite-horizon risk sensitive cost can be approximated by finite $T$-horizon one with the error bound $\theta \beta^{T+1}$ for some constant $\theta > 0$, which is independent of the policy ${\boldsymbol \pi}^{(N,\infty)}$; i.e.,
\begin{align}
\big|W_i^{(N,\infty)}({\boldsymbol \pi}^{(N,\infty)}) - W_i^{(N)}({\boldsymbol \pi}^{(N,\infty)})\big| \leq \theta \beta^{T+1}. \label{uni-app}
\end{align}
Then, the following theorem is a consequence of \eqref{uni-app} and Theorem~\ref{appr-thm}.

\begin{theorem}\label{appr-thm-inf}
For any $\varepsilon>0$, choose $T$ such that $\theta \beta^{T+1} < \frac{\varepsilon}{3}$ and let $N(\frac{\varepsilon}{3})$ be the constant in Theorem~\ref{appr-thm} for the finite horizon $T$. Then, for $N\geq N(\frac{\varepsilon}{3})$, the policy ${\boldsymbol \pi}^{(N,\infty)}$ is an $\varepsilon$-Nash equilibrium for the infinite-horizon risk-sensitive game with $N$ agents, where ${\boldsymbol \pi}^{(N,\infty)} = (\pi^{\infty},\ldots,\pi^{\infty})$,
$$\pi^{\infty} = \big\{\underbrace{\pi_0^*, \ldots,\pi_T^*}_{\text{$T+1$-times}} , \pi_{T+1},\pi_{T+2}, \ldots\big\},$$
$\pi^* = \{\pi_t^*\}_{t=0}^T$ is the policy in the mean-field equilibrium of the $T$-horizon game, and $\{\pi_t\}_{t=T+1}^{\infty}$ is some arbitrary policy.
\end{theorem}


\section{An Example}\label{example}

\def\sS{{\mathsf S}}

In this section, we consider an additive noise model to illustrate our results. In this model, the state and observation dynamics of a generic agent for the infinite-population game are given respectively by
\begin{align}
s(t+1) &= \int_{\sS} f(s(t),u(t),s)  d_t(ds) + g(s(t),u(t))  w(t) \nonumber \\
&\eqqcolon F(s(t),u(t),d_t) + g(s(t),u(t)) w(t) \nonumber \\
\intertext{and}
g(t) &= h(s(t)) +  v(t),  \nonumber
\end{align}
where $s(t) \in \sS$, $g(t) \in \sY$, $u(t) \in \sA$, $w(t) \in \sW$, and $v(t) \in \sV$. Here, we assume that $\sS = \sY = \sW =\sV = \R$, $\sA \subset \R$, and $\{w(t)\}$ and $\{v(t)\}$  are sequences of i.i.d. standard normal random variables independent of each other. The one-stage cost function of a generic agent is given by
\begin{align}
m(s(t),u(t),d_t) = \int_{\sS} b(s(t),u(t),s) \phantom{i} d_t(ds), \nonumber
\end{align}
for some measurable function $b: \sS \times \sA \times \sS \rightarrow [0,\infty)$.

This model is the infinite-population limit of the $N$-agent game model with state and observation dynamics
\begin{align}
s_i^N(t+1) &= \frac{1}{N} \sum_{j=1}^N f(s_i^N(t),u_i^N(t),s_j^N(t)) + g(s_i^N(t),u_i^N(t)) w_i^N(t) \nonumber \\ y_i^N(t) &=  h(s_i^N(t)) + v_i^N(t) \nonumber
\end{align}
and the one-stage cost function
\begin{align}
m(s_i^N(t),u_i^N(t),d_t^{(N)}) &= \frac{1}{N} \sum_{j=1}^N b(s_i^N(t),u_i^N(t),s_j^N(t)). \nonumber
\end{align}

For this model, Assumption~\ref{as1} holds with $v(s) = 1 + s^2$ and $\alpha = \max\{1 + \|f\|^2, L\}$ under the following conditions: (i) $\sA$ is compact, (ii) $b$ is continuous and bounded, (iii) $g$ is continuous, and $f$ is bounded and continuous, (iv) $\sup_{u \in \sA} g^2(s,u) \leq L s^2$ for some $L>0$, (v) $h$ is continuous and bounded. Note that $\|f\|$ is defined as
\begin{align}
\|f\| \coloneqq \sup_{(s,u,s') \in \sS \times \sA \times \sS} |f(s,u,s')|. \nonumber
\end{align}
Moreover, Assumption~\ref{as2}-(a) holds under the following conditions: (vi) $b(s,u,s')$ is (uniformly) Lipschitz in $s'$, (vii) $f(s,u,s')$ is (uniformly) Lipschitz in $s'$, and (viii) $g$ is bounded and $\inf_{(s,u) \in \sS \times \sA} |g(s,u)| > 0$. For the proofs of these facts, we refer the reader to \cite[Section 7]{SaBaRa18}.  

In order to have Assumption~\ref{as2}-(b), we need to assume that $\sA$ is convex. In addition, suppose that $q(ds'|s,a,\mu) = \varrho(s'|s,a,\mu) \nu(ds')$ and $l(dy|s) = \zeta(y|s) \nu(dy)$, where $\nu$ denotes the Lebesgue measure. Assume that both $\varrho$ and $\zeta$ are continuous and bounded, and $\varrho$ and $m$ are strictly convex in $a$. For the justification of Assumption~\ref{as2}-(b) in this case, we refer the reader to Section~\ref{continuity}.

\begin{remark}
Note that Assumption~\ref{as1} also holds for finite models (i.e., $\sS$, $\sA$, and $\sY$ are finite) without any structure on the dynamics of the state and observation if the transition probability and the one-stage cost function are continuous with respect to the mean-field term. Moreover, Assumption~\ref{as2}-(a) holds if the transition probability and the one-stage cost function are Lipschitz continuous with respect to the mean-field term. In finite models, the only missing condition is the existence of deterministic policy in mean-field equilibrium. This can be established if we have the uniqueness condition in (\ref{unique}).
\end{remark}

\section{Conclusion}\label{conc}

This paper has considered discrete-time finite-horizon partially-observed risk-sensitive mean-field games. We have first constructed an equivalent game model whose states are
the state of the original model plus the one-stage costs incurred up to that time. In this new model, the finite-horizon risk-sensitive cost function can be written
in an additive-form as in the risk-neutral case. Then, letting the number of agents go to infinity, we have first established the existence of a mean-field equilibrium in the limiting mean-field game problem. We have then shown that the policy in the mean-field equilibrium constitutes an approximate Nash equilibrium for similarly structured games with a sufficiently large number of agents. Finally, we have extended our results to the case of infinite-horizon cost functions.

\section{Appendix}

\subsection{Continuous and Deterministic Equilibrium Policy}\label{continuity}

\def\sS{{\mathsf S}}

A common way to establish Assumption~\ref{as2}-(b) is as follows. Suppose that, for the measure-flow $\bmu$ in mean-field equilibrium, there exists a unique minimizer $a_z \in \sA$ of
\begin{align}
C_t^{\bmu}(z,\,\cdot\,) + \int_{\sZ} J_{*,t+1}^{\bmu}(z') \eta_t^{\bmu}(dz'|z,\,\cdot\,) = R_t(z,\,\cdot\,), \label{unique}
\end{align}
for each $z \in \sZ$ and for all $t$. In addition, suppose that $F_t:\sZ \times \sA \times \sY \rightarrow \sZ$ in \eqref{eq:non_filtering} is continuous.
Note that uniqueness conditions analogous to \eqref{unique} are quite common in the mean field literature (see, e.g., \cite[Assumption 4]{GoMoSo10}, \cite[Assumption A5]{SeCa16}, \cite[Assumption H5]{HuMaCa06}, \cite[Assumption A9]{SeCa16-3}).

Under the condition of a unique minimizer to \eqref{unique}, one can prove that the policy $\varphi$ in Proposition~\ref{prop1} is deterministic and weakly continuous (see \cite[Remark 5.2]{SaBaRa18}). Indeed, fix any $t\geq0$ and consider the policy $\varphi_t$ at time $t$ in $\varphi$. By the unique minimizer condition \eqref{unique}, we must have $\varphi_t(\,\cdot\,|z) = \delta_{f_t(z)}(\,\cdot\,)$ for some deterministic function $f_t: \sZ\rightarrow \sA$ which minimizes $R_{t}(z,\,\cdot\,)$; that is, $\min_{a \in \sA} R_{t}(z,a) = R_{t}(z,f_t(z))$ for all $z \in \sZ$. If $f_t$ is continuous, then $\varphi_t$ is also weakly continuous. Hence, in order to prove the assertion, it is sufficient to prove that $f_t$ is continuous. Suppose that $z_n \rightarrow z$ in $\sZ$. Note that $l_t(\,\cdot\,) = \min_{a \in \sA} R_{t}(\,\cdot\,,a)$ is continuous. Therefore, every accumulation point of the sequence $\{f_t(z_n)\}_{n\geq1}$ must be a minimizer for $R_{t}(z,\,\cdot\,)$. Since there exists a unique minimizer $f_t(z)$ of $R_{t}(z,\,\cdot\,)$, the set of all accumulation points of $\{f_t(z_n)\}_{n\geq1}$ must be the singleton $\{f_t(z)\}$. This implies that $f_t(z_n)$ converges to $f_t(z)$ since $\sA$ is compact. Hence, $f_t$ is continuous.

Recall that the mean-field equilibrium policy is given by
\begin{align}
\pi_t(\,\cdot\,|g(t)) = \varphi_t(\,\cdot\,|i(g(t))). \nonumber
\end{align}
Hence, $\pi$ is also a deterministic policy as $i$ is a deterministic function. The function $i$ can be generated recursively using $F_t:\sZ \times \sA \times \sY \rightarrow \sZ$ ($t\geq0$) in \eqref{eq:non_filtering} and the policy $\varphi$. Since $F_t$ is continuous for all $t$ and $\varphi$ is also weakly continuous, we can conclude that the mean-field policy $\pi$ is deterministic and weakly continuous. Hence, Assumption~\ref{as2}-(b) holds.

For instance, we can prove the existence of a unique minimizer to \eqref{unique} and the continuity of $F_t$ for all $t$ under the following conditions on the system components. Suppose that $\sS = \R^d$, $\sY = \R^p$, and $\sA \subset \R^m$ is convex. In addition, suppose that $q(ds'|s,a,\mu) = \varrho(s'|s,a,\mu) \nu(ds')$ and $l(dy|s) = \zeta(y|s) \nu(dy)$, where $\nu$ denotes the Lebesgue measure. Assume that both $\varrho$ and $\zeta$ are continuous and bounded, and $\varrho$ and $m$ are strictly convex in $a$, where $m$ is the one-stage cost function of the original problem.
Then we have
$
H_t(dy|z,a) = h_t(y|z,a) \nu(dy), \nonumber
$
where $h_t(y|z,a)$ is given by
\begin{align}
h_t(y|z,a) &= \int_{\sS} \int_{\sS} \zeta(y|s) \varrho(s|s',a,\mu_t) \nu(ds) z_1(ds'), \nonumber
\end{align}
where $z_1(ds') = z(ds' \times [0,L])$. Similarly, we have
\begin{align}
&F_t(dx|z,a,y) = \frac{\int_{\sX} f_t(s|s',a,y) \nu(ds) \otimes \delta_{m'+\beta^t m(s',a,\mu_t)}(dm) z(ds',dm')}{h_t(y|z,a)}, \nonumber
\end{align}
where $f_t(s|s',a,y)$ is given by
$
f_t(s|s',a,y) = \zeta(y|s) \varrho(s|s',a,\mu_t). \nonumber
$
Then, one can prove that $F_t$ is continuous. To show uniqueness of the minimizer to \eqref{unique}, note that
\begin{align}
J_{*,t+1}^{\bmu}(z) &=  \inf_{\varphi \in \Phi} E^{\varphi} \biggl[ \sum_{k=t+1}^{T+1} C_k^{\bmu}(z(k),a(k)) \bigg| z(t+1) = z  \biggr] \nonumber \\
&=\inf_{\pi \in \Pi} E^{\pi} \biggl[ \sum_{k=t+1}^{T+1} c_k(x(k),a(k),\mu_k) \bigg| x(t+1) \sim z  \biggr] \nonumber \\
&= \int_{\sX} V_{*,t+1}(x) z(dx), \nonumber
\end{align}
where
\begin{align}
V_{*,t+1}(x) &=  \inf_{\pi \in \Pi} E^{\pi} \biggl[ \sum_{k=t+1}^{T+1} c_k(x(k),a(k),\mu_k) \bigg| x(t+1) = x  \biggr]. \nonumber
\end{align}
Hence, for any $a \in \sA$, \eqref{unique} can be written as
\begin{align}
&\int_{\sX} c_t(x,a,\mu_t) z(dx) + \int_{\sY} \int_{\sX} V_{*,t+1}(x) F_t(z,a,y)(dx) H_t(dy|z,a) \nonumber \\
&= \hspace{-5pt}\int_{\sX} \hspace{-5pt} c_t(x,a,\mu_t) z(dx)\nonumber \\
&\phantom{xxxxx} + \int_{\sX} \int_{\sS} \hspace{-1pt} V_{*,t+1}(s,m'+\beta^t m(s',a,\mu_t)) \varrho(s|s',a,\mu_t) \nu(ds) z(ds',dm').  \nonumber
\end{align}
Note that
\begin{align}
V_{*,t+1}(s,m) = e^{\lambda m} \inf_{\pi \in \Pi}  e^{\lambda E[ \sum_{k=t+1}^{T} \beta^k m(s(k),u(k),\mu_{1,k}) | s(k) = s] }, \nonumber
\end{align}
and thus $V_{*,t+1}(s,m)$ is strictly convex in $m$. Since $m$ and $\varrho$ are strictly convex in $a$, the last expression is also strictly convex in $a$. Hence, there exists a unique minimizer $a_z \in \sA$ for \eqref{unique}.

\subsection{Proof of Proposition~\ref{prop9}}\label{app3}

We prove the result by induction on $t$. The claim trivially holds for $t=0$ as ${\cal L}(\tb_1^N(0)) = {\cal L}(\hb^N(0)) = \lambda_0$ for all $N\geq1$. Suppose that the claim holds for $t$ and consider $t+1$. Set $\sup_{N\geq1} \|g_N\| \eqqcolon L <\infty$ and define
\begin{align}
T_N(b,\Delta) \coloneqq \int_{\sA \times \sB_{t+1}} g_N(b') P_t(db'|b,a,\Delta) \tpi_t^{(N)}(da|b). \nonumber
\end{align}
We can write
\begin{align}
&\bigl| {\cal L}(\tb_1^N(t+1))(g_N) - {\cal L}(\hb^N(t+1))(g_N) \bigr| \nonumber \\
&\phantom{xxxxx}=\biggl| \int_{\sB_t \times \P(\sB_t)} T_N(b,\Delta) {\cal L}(\tb_1^N(t),\tilde{\Delta}_t^{(N)})(db,d\Delta) \nonumber \\
&\phantom{xxxxxxxxxxxxxxxxxxxxxxxx}- \int_{\sB_t \times \P(\sB_t)} T_N(b,\Delta) {\cal L}(\hb^N(t),\delta_{\Delta_t})(db,d\Delta) \biggr|  \nonumber \\
&\phantom{xxxxx}= \bigl| {\cal L}(\tb_1^N(t),\tilde{\Delta}_t^{(N)})(T_N) - {\cal L}(\hb^N(t),\delta_{\Delta_t})(T_N) \bigr|. \nonumber 
\end{align}
Note that, for any $b \in \sB_t$ and $(\Delta,\Delta') \in \P(\sB_t)^2$, we have
\begin{align}
&|T_N(b,\Delta) - T_N(b,\Delta')| \leq \omega_g\bigl(\omega_m(d_{BL}(\Delta_1,\Delta'_1))\bigr) +  L \omega_q(d_{BL}(\Delta_1,\Delta'_1)), \nonumber
\end{align}
where $\Delta_1$ is the marginal distribution of $s$ under $\Delta$ (recall that $b=(s,m,y^t)$). Hence the family $\{T_N(b,\,\cdot\,): b \in \sB_t, N\geq1\}$ is uniformly bounded and equi-continuous. Moreover, for any $\Delta \in \P(\sB_t)$, we have
\begin{align}
\omega_{T,\Delta}(r) &\coloneqq \sup_{s,y^t} \sup_{N\geq1} \sup_{\substack{m,m' \\ |m - m'| \leq r}} |T_N(s,m,y^t,\Delta) - T_N(s,m',y^t,\Delta)| \nonumber \\
&\leq \sup_{s,y^t} \sup_{N\geq1} \sup_{\substack{m,m' \\ |m - m'| \leq r}} \omega_g(|m-m'|) = \omega_g(r). \nonumber
\end{align}
Hence, $\omega_{T,\Delta}(r) \rightarrow 0$ as $r \rightarrow 0$. Therefore, $\{T_N\} \subset C_b(\sB_t\times\P(\sB_t))$ is a sequence of functions such that the family $\bigl\{T_N(b,\,\cdot\,): b \in \sB_t, N\geq1)\bigr\}$ is equi-continuous, $\sup_{N\geq1}\|T_N\|<\infty$, and $\omega_{T,\Delta}(r) \rightarrow 0$ as $r \rightarrow 0$ for any $\Delta \in \P(\sB_t)$.
We now prove that
\begin{align}
\lim_{N\rightarrow\infty} \bigl| {\cal L}(\tb_1^N(t),\tilde{\Delta}_t^{(N)})(T_N) - {\cal L}(\hb^N(t),\delta_{\Delta_t})(T_N) \bigr| = 0, \label{neweq2}
\end{align}
which would then complete the proof. Indeed, we have
\begin{align}
&\bigl| {\cal L}(\tb_1^N(t),\tilde{\Delta}_t^{(N)})(T_N) - {\cal L}(\hb^N(t),\delta_{\Delta_t})(T_N) \bigr| \nonumber \\
&\phantom{xxxxxx}\leq \biggl| \int_{\sB_t \times \P(\sB_t)} T_N(b,\Delta) {\cal L}(\tb_1^N(t),\tilde{\Delta}_t^{(N)})(db,d\Delta) \nonumber \\
&\phantom{xxxxxxxxxxxxxxxxxxx}- \int_{\sB_t \times \P(\sB_t)} T_N(b,\Delta) {\cal L}(\tb_1^N(t),\delta_{\Delta_t})(db,d\Delta) \biggr| \nonumber \\
&\phantom{xxxxxx}+ \biggl| \int_{\sB_t \times \P(\sB_t)} T_N(b,\Delta) {\cal L}(\tb_1^N(t),\delta_{\Delta_t})(db,d\Delta)  \nonumber \\
&\phantom{xxxxxxxxxxxxxxxxxxx}- \int_{\sB_t \times \P(\sB_t)} T_N(b,\Delta) {\cal L}(\hb^N(t),\delta_{\Delta_t})(db,d\Delta) \biggr|. \label{eq12}
\end{align}
First, note that since the family $\{T_N(\,\cdot\,,\Delta_t)\}_{N\geq1} \subset C_b(\sB_t)$ satisfies the hypothesis of the proposition and the proposition is true for $t$, by induction hypothesis, we have
\begin{align}
\lim_{N\rightarrow\infty} \biggl| \int_{\sB_t} T_N(b,\Delta_t) {\cal L}(\tb_1^N(t))(db) - \int_{\sB_t} T_N(b,\Delta_t) {\cal L}(\hb^N(t))(db) \biggr|=0. \nonumber
\end{align}
Hence, the second term in (\ref{eq12}) converges to zero as $N\rightarrow\infty$.

Now, let us consider the first term in (\ref{eq12}). To that end, define ${\cal F} \coloneqq \bigl\{T_N(b,\,\cdot\,): b \in \sB_t, N\geq1)\bigr\}$. Note that ${\cal F}$ is a uniformly bounded and equi-continuous family of functions on $\P(\sB_t)$, and therefore
\begin{align}
\lim_{N\rightarrow\infty} E\biggl[ \sup_{F \in {\cal F}} \bigl| F(\tilde{\Delta}_t^{(N)}) - F(\Delta_t) \bigr| \biggr] = 0 \nonumber
\end{align}
as ${\cal L}(\tilde{\Delta}_t^{(N)}) \rightarrow {\cal L}(\Delta_t)$ weakly. Then, we have
\begin{align}
&\lim_{N\rightarrow\infty}\biggl| \int_{\sB_t \times \P(\sB_t)} T_N(b,\Delta) {\cal L}(\tb_1^N(t),\tilde{\Delta}_t^{(N)})(db,d\Delta) \nonumber \\
&\phantom{xxxxxxxxxxxxxxxxxx}- \int_{\sB_t \times \P(\sB_t)} T_N(b,\Delta) {\cal L}(\tb_1^N(t),\delta_{\Delta_t})(db,d\Delta) \biggr| \nonumber \\
&\leq \lim_{N\rightarrow\infty} \int_{\sB_t} \biggl| \int_{\P(\sB_t)} T_N(b,\Delta) {\cal L}(\tilde{\Delta}_t^{(N)}|\tb_1^N(t))(d\Delta|b) \nonumber \\
&\phantom{xxxxxxxxxxxxxxxxxx}- \int_{\P(\sB_t)} T_N(b,\Delta) {\cal L}(\delta_{\Delta_t})(d\Delta) \biggr| {\cal L}(\tb_1^N(t))(db) \nonumber \\
&\leq \lim_{N\rightarrow\infty} E\biggl[ E\biggl[ \bigl|T_N(\tb_1^N(t),\tilde{\Delta}_t^{(N)}) - T_N(\tb_1^N(t),\Delta_t) \bigl| \biggl| \tb_1^N(t) \biggr] \biggr] \nonumber \\
&\leq \lim_{N\rightarrow\infty} E\biggl[ \sup_{F \in {\cal F}} \bigl|F(\tilde{\Delta}_t^{(N)}) - F(\Delta_t) \bigl| \biggr] \nonumber \\
&= 0. \nonumber
\end{align}
This completes the proof.

\begin{acknowledgements}
This work was partly supported by The Scientific and Technological Research Council of Turkey (T\"{U}B\.{I}TAK) B\.{I}DEB 2232 Research Grant.
\end{acknowledgements}

\end{document}